\documentclass[11pt]{article}

\usepackage{amsfonts}
\usepackage{amsmath}
\usepackage{amssymb}
\usepackage{amsthm}

\usepackage{graphicx}

\usepackage{psfrag}
\theoremstyle{plain}
\newtheorem{theorem}{Theorem}
\newtheorem{proposition}[theorem]{Proposition}
\newtheorem{definition}[theorem]{Definition}
\newtheorem{lemma}[theorem]{Lemma}

\theoremstyle{definition}

\newtheorem{remark}[theorem]{Remark}

\numberwithin{equation}{section}
\numberwithin{theorem}{section}

\DeclareMathOperator{\Tr}{tr}

\DeclareMathOperator{\diag}{diag}

\newcommand{\bs}[1]{{\boldsymbol{#1}}}

\newcommand{\R}{\mathbf{R}}
\newcommand{\N}{\mathbf{N}}

\begin{document}
%The heading
\title{{Linear algebra of the permutation invariant Crow--Kimura model of prebiotic evolution}}

\author{Alexander S. Bratus$^{1,2}$, Artem S. Novozhilov$^{{3},}$\footnote{Corresponding author: artem.novozhilov@ndsu.edu}\,\,, Yuri S. Semenov$^{2}$\\[3mm]
\textit{\normalsize $^\textrm{\emph{1}}$Faculty of Computational Mathematics and Cybernetics,}\\[-1mm]
\textit{\normalsize Lomonosov Moscow State University, Moscow 119992, Russia}\\[2mm]
\textit{\normalsize $^\textrm{\emph{2}}$Applied Mathematics--1, Moscow State University of Railway Engineering,}\\[-1mm]\textit{\normalsize Moscow 127994, Russia}\\[2mm]
\textit{\normalsize $^\textrm{\emph{3}}$Department of Mathematics, North Dakota State University, Fargo, ND 58108, USA}}

\date{}

\maketitle

%Abstract goes here
\begin{abstract}
A particular case of the famous quasispecies model --- the Crow--Kimura model with a permutation invariant fitness landscape --- is investigated. Using the fact that the mutation matrix in the case of a permutation invariant fitness landscape has a special tridiagonal form, a change of the basis is suggested such that in the new coordinates a number of analytical results can be obtained. In particular, using the eigenvectors of the mutation matrix as the new basis, we show that the quasispecies distribution approaches a binomial one and give simple estimates for the speed of convergence. Another consequence of the suggested approach is a parametric solution to the system of equations determining the quasispecies. Using this parametric solution we show that our approach leads to exact asymptotic results in some cases, which are not covered by the existing methods. In particular, we are able to present not only the limit behavior of the leading eigenvalue (mean population fitness), but also the exact formulas for the limit quasispecies eigenvector for special cases. For instance, this eigenvector has a geometric distribution in the case of the classical single peaked fitness landscape. On the biological side, we propose a mathematical definition, based on the closeness of the quasispecies to the binomial distribution, which can be used as an operational definition of the notorious error threshold. Using this definition, we suggest two approximate formulas to estimate the critical mutation rate after which the quasispecies delocalization occurs.
\end{abstract}

\section{Introduction}
In 1971 Manfred Eigen published a very influential paper \cite{eigen1971sma}, in which he proposed a far-reaching theory for the origin of life, that included biological, physical, chemical, mathematical, and information-theoretic aspects. The details of this theory, dubbed as \textit{quasispecies theory}, were developed in various papers by Eigen himself and by many others; a comprehensive summary can be found in \cite{eigen1989mcc}, a shorter version with the basic details and conclusions is \cite{eigen1988mqs}, for a more recent review see \cite{jainkrug2007}. In the present paper we are going to discuss only the aspects pertaining to the mathematical side of the story, therefore, to make the text self contained, we start with a formulation of several mathematical models (systems of ordinary differential or recurrence equations), each of which was called a quasispecies model in the literature. These models bear a close resemblance to each other and produce similar results; it is important, however, from the very beginning to state precisely which particular model is analyzed.

We begin with a more natural approach with discrete time steps. Imagine a population of sequences (individuals) of the length $N$, each of which is composed of a two letter alphabet (0s and 1s, for instance); we have $2^N$ different sequences, which we denote $\sigma_i,\,i=1,\ldots,2^N$. We also use the notation $n(\sigma_i,t)$ for the total number of the sequences of type $\sigma_i$ at the time moment $t$. Assume that sequence $\sigma_i$ begets $w(\sigma_i)$ offspring on average; the reproduction is error-prone, so that upon any reproduction event the probability that sequence $\sigma_j$ produces sequence of type $\sigma_i$ is $s_{ij}=s (\sigma_j\to\sigma_i)$ and $s_{jj}=1-\sum_{i=1,i\neq j}^{2^N}{s_{ij}}$. Simple bookkeeping for the population size at the next time moment leads to
\begin{equation}\label{eq0:1}
    n(\sigma_i,t+1)=\sum_{j=1}^{2^N}w(\sigma_j)s_{ij}n(\sigma_j,t),\quad i=1,\ldots,2^N.
\end{equation}
For many cases it is more convenient (see, e.g., \cite{karev2010} for a discussion) to consider the system for the corresponding frequencies
$$p(\sigma_i,t)=\frac{n(\sigma_i,t)}{\sum_{j=1}^{2^N}n(\sigma_j,t)}\,,$$ which takes the form
\begin{equation}\label{eq0:2}
    \bs{p}(t+1)=\frac{\bs{SW}\bs{p}(t)}{\overline{w}(t)}\,,
\end{equation}
where $\bs{p}(t)=\bigl(p(\sigma_1,t),\ldots,p(\sigma_{2^N},t)\bigr)^\top\in \R^{2^N}$ for any $t\in \mathbf{N}$, $^\top$ denotes transposition, matrix $\bs{W}$ is a diagonal matrix, $\bs{W}=\diag\bigl(w(\sigma_1),\ldots,w(\sigma_{2^N})\bigr)$, which is commonly referred as \textit{the fitness landscape}, $\bs{S}=(s_{ij})_{2^N\times 2^N}$, and $\overline{w}(t)$ is the mean (Wrightian) fitness of the population defined~by
\begin{equation}\label{eq0:3}
    \overline{w}(t)=\sum_{j=1}^{2^N}w(\sigma_j)p(\sigma_j,t).
\end{equation}
To further specify the model it is usual to make an assumption that the mutation probability is constant, say $s$, per site per replication and mutations occur independently of each other. Then the mutation matrix $\bs{S}$ can be specified as a matrix depending only on one parameter~$s$:
\begin{equation}\label{eq0:4}
    s_{ij}=s(\sigma_j\to\sigma_i)=s^{H(\sigma_i,\sigma_j)}(1-s)^{N-H(\sigma_i,\sigma_j)},
\end{equation}
where $H(\sigma_i,\sigma_j)$ denotes the standard Hamming distance between sequences $\sigma_i$ and $\sigma_j$.

Summarizing, the model \eqref{eq0:2} is defined by the sequence length $N$, the mutation probability per site per sequence per replication event $s$, and by the (Wrightian) fitness landscape $\bs{W}$.

Majority of biological systems do not possess the property of giving birth at the same time moment, therefore it is desirable to formulate an analogue of \eqref{eq0:2} in continuous time. The phenomenological model employed by Eigen et al. \cite{eigen1989mcc,eigen1988mqs} takes the form of the system of ordinary differential equations
\begin{equation}\label{eq0:5}
    \bs{\dot{p}}(t)=\bs{SW}\bs{p}(t)-\overline{w}(t)\bs{p}(t),
\end{equation}
where all the notations as before, however, the constants $w(\sigma_i)$ should be now interpreted as Malthusian fitnesses, since in the statement of the model Eigen et al. assumed that $w(\sigma_i)=b(\sigma_i)-d(\sigma_i)$, where $b(\sigma_i)$ and $d(\sigma_i)$ are the birth and death \textit{rates} respectively. Note that the stationary point of \eqref{eq0:5} coincides precisely with the stationary point of \eqref{eq0:2}. We are not aware of any mechanical derivation of the model \eqref{eq0:5} as a limit of some system when time unit approaches zero. The usual strategy of the Markov process theory fails in this particular case because the nature of the model \eqref{eq0:5} implies that several elementary events (births) can occur at the same time. These remarks notwithstanding, system \eqref{eq0:5} plus additional assumption \eqref{eq0:4} is what is usually referred to as the quasispecies model in the literature on the prebiotic evolution.

Another approach to continuous time is to start with the assumptions that the birth events and mutations are separated on our time scale, i.e., the mutations occur during the time life of the sequences, and birth events are error free. In this case, denoting $\mu_{ij}=\mu(\sigma_j\to\sigma_i)$ the mutation rates of sequence $\sigma_j$ to $\sigma_i$, $\mu_{jj}=-\sum_{i\neq j}\mu_{ij}$, and $m(\sigma_i)=b(\sigma_i)-d(\sigma_j)$ the Malthusian fitness of the sequence $\sigma_i$, the standard bookkeeping leads to the system
\begin{equation}\label{eq0:6}
    \bs{\dot{p}}(t)=\bigl(\bs{M}-\overline{m}(t)\bs{I}\bigr)\bs{p}(t)+\bs{\mathcal M}\bs{p}(t),
\end{equation}
where $\bs{M}=\diag\bigl(m(\sigma_1),\ldots,m(\sigma_{2^N})\bigr)$ is the fitness landscape, $\bs{\mathcal M}=(\mu_{ij})_{2^N\times 2^N}$ is the mutation matrix, $\bs{I}$ is the identity matrix, $\overline{m}(t)=\sum_{j=1}^{2^N}m(\sigma_j)p(\sigma_j,t)$ is the mean population fitness.

Model \eqref{eq0:6} has a straightforward mechanical derivation in terms of elementary processes, however it is not quite clear how realistic is the assumption about separation of the reproduction events and mutations. To supplement model \eqref{eq0:6} we can assume that if the mutation rate per site per sequence per replication event is constant, say $\mu$, then the mutation matrix $\bs{\mathcal M}$ can be written as
\begin{equation}\label{eq0:7}
    \mu_{ij}=\begin{cases}\mu,&\mbox{if }H(\sigma_i,\sigma_j)=1,\\
    0,&\mbox{if }H(\sigma_i,\sigma_j)>1,\\
    -N\mu,&\mbox{if }H(\sigma_i,\sigma_j)=0.
    \end{cases}
\end{equation}
Model \eqref{eq0:6}--\eqref{eq0:7} was advocated by E. Baake and co-authors \cite{baake1999} and was often dubbed as a paramuse model, due to parallel mutation selection scheme. This same model actually was studied before the Eigen quasispecies theory as the model of mutation--selection balance in a haploid asexual population, and was given some consideration in Chapter 6 in the textbook by Crow and Kimura \cite{crow1970introduction}, therefore model \eqref{eq0:6} is also often called the Crow--Kimura quasispecies model. It is interesting to note that the fact that the models studied within the framework of quasispecies theory has been long an object of study in theoretical population genetics was stressed only relatively recently (e.g., \cite{baake1999,wilke2005quasispecies}, an intricate mathematical theory of mutation--selection balance is dealt with in \cite{burger2000mathematical}).

Another quite natural way to deduce model \eqref{eq0:6} is to start with \eqref{eq0:2} and consider its limit for small generation time \cite{hofbauer1985selection}. In this way model \eqref{eq0:6} is the limit when generation time tends to zero. Moreover, this limit procedure helps relate the Wrightian and Malthusian fitnesses:
\begin{equation}\label{eq0:8}
    w(\sigma_i)=\exp\bigl(m(\sigma_i)\Delta t\bigr)\approx 1+m(\sigma_i)\Delta t,\quad \Delta t\to0,
\end{equation}
and the mutations probabilities and corresponding mutation rates
\begin{equation}\label{eq0:9}
    s_{ij}\approx \delta_{ij}+\Delta t\,\mu_{ij},\quad \Delta t\to 0,
\end{equation}
where $\delta_{ij}$ is the Kronecker delta.

All three models \eqref{eq0:2}, \eqref{eq0:5}, \eqref{eq0:6} possess similar properties that can be summarized as follows:
\begin{enumerate}
\item Models \eqref{eq0:2}, \eqref{eq0:5}, and \eqref{eq0:6}, under some mild technical assumptions on the used matrices, admit the only globally stable stationary point $\lim_{t\to\infty}\bs{p}(t)=\bs{p}\in\R^{2^N}$, which was called \textit{the quasispecies} by Eigen.
\item For some fitness landscapes (and in particular for the so-called \textit{single} or \textit{sharply peaked fitness landscape} when $w(\sigma_1)>w(\sigma_2)=\ldots=w(\sigma_{2^N})$) there exists a sharp transition, called the \textit{error threshold}, which separates the regime where the fittest sequence (sometimes called the \textit{master sequence}) dominates the population, from the regime where the distribution of the frequencies of different types of sequences is close to uniform.
\end{enumerate}
The first fact is a simple consequence of the form of the matrix $\bs{SW}$ (or $\bs{M}+\bs{\mathcal M}$) and the Perron--Frobenius theorem: Recall that all three systems \eqref{eq0:2}, \eqref{eq0:5}, and \eqref{eq0:6} are \textit{almost} linear (cf. \eqref{eq0:1} and \eqref{eq0:2}). The second observation of the quasispecies theory is quite difficult to discuss mathematically, because in many cases we have only a qualitative definition of the error threshold, the literature survey shows that in many instances nonequivalent definitions are used~\cite{tejero2011relationship}.

It is not our goal to discuss at this point various aspects of the error threshold definition, we defer it to the last section of our manuscript. Here we would like to mention that very different approaches were used to study the quasispecies systems. The linear nature of the problem notwithstanding, the fact that the dimension of the systems is $2^N$ precludes any direct calculation for even moderate values of $N$.

A first direction to simplify one of the quasispecies system is to note that if the fitness landscape is \textit{permutation invariant}, i.e., depends not on the sequence itself, but rather on the sequence composition such that permutations of zeros and ones in this sequence do not change any sequence's fitness, the dimension of the problem is reduced significantly: from $2^N$ different types of sequences to $N+1$ classes of sequences, each of which contains $\binom{N}{i}$ different types of sequences, $i=0,\ldots,N$. Using this representation (for the models \eqref{eq0:2} and \eqref{eq0:5} the explicit form for the matrix $\bs{S}$ is given in, e.g., \cite{nowak1989error}; for \eqref{eq0:6} and for the matrix $\bs{\mathcal M}$ see the next section), extensive numerical results were obtained, e.g., \cite{nowak1989error,schuster1988stationary,swetina1982self}, that also confirmed various heuristic approximations.

Another fruitful approach to tackle the quasispecies model is to apply methods from statistical physics. In particular, it was shown \cite{leuthausser1986exact,leuthausser1987statistical} that the Eigen model \eqref{eq0:2} is equivalent to two-dimensional classical Ising model in statistical physics. This correspondence generated a stream of papers analyzing one or another incarnation of the Eigen model, see, e.g., \cite{baake2001mutation} for some details. In this last paper it is also shown that the Crow--Kimura model is equivalent to the so-called Ising quantum spin chain. In particular, a number of ``exact'' results appeared in the literature, e.g., \cite{baake1997ising,baake2001mutation,galluccio1997exact,saakian2006ese}. The word ``exact'' is in quotes because the meaning was borrowed from statistical physics and means that it is possible to obtain some analytical solutions for the quantities of interest under special scalings and limit procedures. It is well known that for the exact solutions of the quasispecies models in the usual sense it is necessary to consider either some approximations (e.g., to prohibit backward mutations to the fittest sequence, for a representative example see \cite{wiehe1997model}) or restrict the attention to the so-called Fujiyama fitness landscape. In the latter case the fitness ladnscape is defined as a multiplicative one for \eqref{eq0:2} or \eqref{eq0:5}, when the fitness of the class $j$ is given by $r^j$ for some constant $r$, or additive fitness landscape for \eqref{eq0:6}, when the fitness of the class $j$ is $rj$ for some constant $r$.

The special structure of the matrices $\bs{SW}$ or $\bs{M}+\bs{\mathcal M}$ allowed to obtain some results, especially those concerning the spectral properties of these matrices, using the standard methods of linear algebra (see \cite{garcia2002linear} for the model \eqref{eq0:6} and \cite{rumschitzki1987spectral} for \eqref{eq0:2} and \eqref{eq0:5}). We pursue a similar rout in our text, focusing our attention initially on finite dimensional problems, i.e., without invoking any limit procedures. To make progress, we use an additional simplifying assumption that the only fitness landscapes we consider are permutation invariant for the model \eqref{eq0:6}. This assumption implies that the mutation matrix is tridiagonal and allows obtaining analytical insights about the behavior of both the mean population fitness and the quasispecies distribution depending on the mutation parameter $\mu$ and given fitness landscape $\bs M$.

We stress that the main results we present in the first part of our manuscript are exact in the usual strict sense. The utility of such approach can be argued as follows. In \cite{Hermisson2002} a general result, which is formulated in the form of a maximal principle (see Section \ref{example:1} for details), was obtained (and a significant generalization is presented in \cite{Baake2007}) about quasispecies models with permutation invariant fitness landscapes. The maximum principle predicts, in the limit of the infinite sequence length, the mean population fitness (the leading eigenvalue) given the fitness landscape and mutational scheme, provided some technical conditions on the fitness landscape and the mutation rates are satisfied. While the results found in \cite{Hermisson2002} are quite general, they still require some scalings, limits, approximations, and technical requirements on the fitness landscape and mutational rates. In Section \ref{example:1} we present a natural example, formulated along the lines of the single peaked fitness landscape, to compare our approach, which may not lead to a concise closed form expressions in the case of finite sequence length, with the maximum principle. First, the maximum principle cannot be applied to our example due to some technical conditions. Second, a formal application of the maximum principle to our example leads to erroneous conclusions, thus confirming that at least for some cases our exact approach works whereas known approaches fail. We also stress that the solutions we put forward in this manuscript allow to obtain not only the mean population fitness, but also the quasispecies distribution. In particular, in Section \ref{example:2} we present an exact asymptotic solution to the quasispecies model with the single peaked fitness landscape, which include both the leading eigenvalue and the corresponding eigenvector. This asymptotic solution was not known for a while despite a number of attempts to tackle this system (e.g., \cite{galluccio1997exact,rumschitzki1987spectral}), and was written down without convergence proof in \cite{saakian2004ese}. We provide a mathematically rigorous derivation of this result with an explicit estimate of the speed of convergence.

To summarize, our main contribution in the present manuscript is trifold. First, it is a purely algebraic approach to the well known Crow--Kimura model, which is based mainly on the reformulation of the original problem in the coordinates of the basis composed of the eigenvectors of the mutation matrix (we also use a similar approach to analyze \eqref{eq0:5} elsewhere \cite{semenov2014}). Using this approach we immediately obtain estimates of the speed of convergence of the quasispecies to the binomial distribution and also write down a parametric solution to the basic eigenvalue problem. Second, to illustrate the generality of the suggested approach we give an example of a system, which is not covered by the existing methods, where our approach allows explicit exact calculations. We also provide a simple exact solution for the asymptotic quasispecies distribution in the case of the infinite genome length for the single peaked fitness landscape. Third, using the analytical insights from the previous parts, we propose an operational definition of the notorious error threshold and present two approximate formulas for the critical mutation rate in the biologically relevant case of finite sequence length.

%From the biological standpoint, Eigen's quasispecies model meant to provide a theoretical (mathematical) framework to address the important question of possible routs of the origin of life. It was understood later that the approach he and his colleagues undertook is mathematically equivalent to the classical population genetics models studying the mutation--selection balance in one-locus $K$-allele haploid populations (or multilocus models with complete linkage, or diploid models with no dominance). Moreover, the fact that much of the evolution of RNA viruses can be cast in terms of quasispecies theory (see, e.g., \cite{Domingo2012,Lauring2010}, but also \cite{Holmes2010} for an alternative critical opinion) brought the biological significance of the study of the quasispecies model to a new level. Any biological system that deals with a population of replicators under high mutation rates can be modeled (as a first approximation) by the deterministic system of ordinary differential equations, that is the central object of analysis in the present text. The role of exact solutions was emphasized many times in the literature (see, e.g., \cite{baake2001mutation}). Due to the fact that we still lack a full understanding of the fitness landscape for even simplest in vitro model populations, any solution for any particular fitness landscape is of great value for the whole quasispecies theory, especially taking into account the rarity of fitness landscapes allowing transparent analytical results. Therefore, the

The rest of the paper is organized as follows. The notations and the permutation invariant Crow--Kimura model are introduced in detail in Section 2. Section 3, albeit technical in nature, provides the main mathematical tools for the subsequent analysis. In particular, the special structure of the mutation matrix allows full characterization of the spectral properties of this matrix together with additional information. As an illustration of the suggested approach, in Section 4 we find estimates of the speed of convergence of the quasispecies distribution to the binomial distribution, which occurs for any fitness landscape. This last fact implies that the error threshold, defined as the mutation rate after which the distribution of different classes of sequences is close to binomial, inherent in this weaker sense to any quasispecies model. Another illustration of the suggested approach is a parametric solution of the system of equations determining the quasispecies in Section 5. Applications of the found parametric formulas lead in Section 6 to analysis of two particular fitness landscapes. We show in Section \ref{example:1} by way of an example that our exact formulas can be applied in situations where the known results cannot be used, and in Section \ref{example:2} we derive an exact asymptotic solution to the quasispecies model with the single peaked fitness landscape. An important question is how to define the error threshold and approximate the critical mutation rate for it, such that the model predictions can be used for studying biological populations of, e.g., viruses. In Section 7 we argue that a natural mathematical definition of the error threshold is the closeness of the quasispecies to the binomial distribution, and present heuristic approximate formulas, based on our definition, to calculate the critical mutation rate. In Section 8, speculative in nature, we discuss the notion of the error threshold, present yet another approximate formula for it, which is based on the parametric solution from Section 6 and simple geometric ideas. Finally, Appendices contain additional results, which extend the presentation of the main text.

\section{The Crow--Kimura model of prebiotic evolution with a permutation invariant fitness landscape}
We start by precisely formulating the mathematical model we study. This is especially important because in the literature similar but different in details models can be called the quasispecies or Crow--Kimura models, see the discussion in the previous section.

The Crow--Kimura model of prebiotic evolution for a permutation invariant fitness landscape for $N+1$ classes of sequences (hence each class contains $\binom{N}{i}$ different types of sequences, $i=0,\ldots,N$) takes the form
\begin{equation}\label{eq1:1}
    \bs{\dot p}(t)=(\bs{M}+\mu \bs{Q})\bs p(t)-{\overline{m}(t)}\bs{p}(t).
\end{equation}
Here
\begin{equation}\label{eq1:2}
\bs{p}(t)=\bigl(p_0(t),\ldots,p_N(t)\bigr)^\top\in S_{N+1}=\{\bs{p}\in\mathbf{R}^{N+1}\mid \bs{p}\geq 0,\,\sum_{i=0}^Np_i=1\}
\end{equation}
is the vector of frequencies of the classes of sequences normalized to belong to simplex $S_{N+1}$; $\bs{M}=\diag(m_0,\ldots,m_N)$ is a diagonal matrix that determines the (Malthusian) fitnesses of different classes of sequences, the corresponding vector $\bs{m}=(m_0,\ldots,m_N)^\top\in \R^{N+1}$; $\bs{Q}=\bs{Q}_N$ is the mutation matrix that has the form
\begin{equation}\label{eq1:3}
    \bs{Q}_N=\begin{bmatrix}
               -N & 1 & 0 & 0 & \ldots & \ldots & 0 \\
               N & -N & 2 & 0 & \ldots & \ldots & 0 \\
               0 & N-1 & -N & 3 & \ldots & \ldots & 0 \\
               0 & 0 & N-2 & -N & \ldots & \ldots & 0 \\
               \ldots & \ldots & \ldots & \ldots & \ldots & \ldots & 0 \\
               0 & 0 & \ldots & \ldots& 2 & -N & N \\
               0 & 0 & \ldots & \ldots & 0& 1 & -N \\
             \end{bmatrix};
\end{equation}
$\mu$ is the mutation rate per site per sequence per time unit; finally,
\begin{equation}\label{eq1:4}
\overline{m}(t)=\bs{m}\cdot \bs{p}(t)=\sum_{i=0}^Nm_ip_i(t)
\end{equation}
is the mean population fitness, $\cdot$ denotes the standard dot product in $\R^{N+1}$.

The form of the mutation matrix \eqref{eq1:3} can be understood as follows: The mutations from class $k$ to class $k+1$ occur with the rate $\mu (N-k)$, since we need to mutate one of the sites among $N-k$ sites with 0s, and the mutations from class $k$ to $k-1$ occur with the rate $k\mu$, since we need to mutate one of $k$ sites with 1s.

Note that equation \eqref{eq1:1} will not change if the fitnesses are scaled as $m_i+\tilde{m},\,i=0,\ldots, N$ for some constant $\tilde{m}$, therefore it is always possible to consider such fitness landscapes that $\min_i\{ m_i\}=0$ (or any other convenient scaling).

The asymptotic behavior of the solutions to \eqref{eq1:1} is determined by the stationary point $\bs{{p}}=\lim_{t\to\infty}\bs{p}(t)$, which is a solution to the system
\begin{equation}\label{eq1:5}
   (\bs{M}+\mu \bs{Q})\bs p={\overline{m}}\,\bs{p},
\end{equation}
where
\begin{equation}\label{eq1:6}
\overline{m}=\bs{m}\cdot \bs{p}.
\end{equation}
By the Perron--Frobenius theorem and standard arguments it follows that system \eqref{eq1:5} has a unique positive solution $\bs{p}\in S_{N+1},\,p_i>0,\,i=0,\ldots,N$, which is the right eigenvector of the matrix $\bs{M}+\mu \bs{Q}$ corresponding to the simple real dominant eigenvalue $\lambda=\overline{m}$. This eigenvector $\bs{p}$ was called by Manfred Eigen the \textit{quasispecies} \cite{eigen1971sma}, hence the name for the theory; the fact that $\lambda$ dominates all other eigenvalues guarantees that $\bs{p}$ is globally stable for system \eqref{eq1:1}.

Both the dominant eigenvalue $\lambda=\overline{m}$ and the stationary solution $\bs{p}$ to \eqref{eq1:5}  depend on the mutation rate $\mu$, provided that the sequence length $N$ and the fitness landscape $\bs{M}$ are fixed. If the mutation rate $\mu$ changes, the stationary solution $\bs{p}$ to \eqref{eq1:5} also changes. In the following we use the notation $\overline{m}(\mu)$ and $\bs p(\mu)$ to emphasize this dependence and stress that with \textit{no} exceptions we deal only with the stationary problem in the rest of the text.  Our task in the present manuscript is to infer analytical insights on the behavior of $\overline{m}(\mu)$ and $\bs{p}(\mu)$ depending on $\mu$ given $\bs M$ and $N$.

\section{Properties of the mutation matrix}
The results of this section, though technical in nature, are central for the following analysis. It turns out that the structure of the matrix $\bs{Q}_N$ allows one to provide a full characterization of its eigenvalues and eigenvectors, which are convenient to use as a different basis for our basic eigenvalue problem \eqref{eq1:5}. Matrix $\bs{Q}_N$ belongs to the family of matrices used in Ehrenfest urn models \cite{karlin1965ehrenfest}, and its eigenvalues and eigenvectors can be expressed in terms of the Krawtchouk system of orthogonal polynomials (see, e.g., \cite{karlin1965ehrenfest}). However, since this representation is not required for our purposes, and to make the paper self-contained, we provide a simple closed expression for the eigenvectors of $\bs{Q}_N$.

The key idea in the proof of the following proposition is to write a linear differential operator $\mathcal Q_N$ acting on the space of polynomials of degree less or equal $N$ whose matrix representation is given by~$\bs Q_N$.

\begin{proposition}\label{lem3:1} For the matrix $\bs{Q}=\bs{Q}_N$ defined by \eqref{eq1:3} the following holds:
\begin{enumerate}
\item The eigenvalues of $\bs{Q}_N$ are simple (all have algebraic multiplicities one) and given by
\begin{equation}\label{eq3:1}
q_k=-2k,\quad k=0,\ldots,N.
\end{equation}

\item Let $\bs{v}_k^\top=(c_{0k},\ldots,c_{Nk})$ be the right eigenvector of $\bs{Q}_N$ corresponding to $q_k$ and normalized such that $c_{0k}=1$, $\bs{C}=\bs{C}_N=(c_{ik})_{(N+1)\times (N+1)}$ be the matrix composed of $\bs{v}_k$ ($\bs{v}_k$ is the $k$-th column of $\bs{C}_N$). Then the generating function for the elements of the $k$-th column has the form
\begin{equation}\label{eq3:2}
    P_k(t)=\sum_{i=0}^N c_{ik}\, t^i=(1-t)^k(1+t)^{N-k},\quad k=0,\ldots,N.
\end{equation}
\item $\bs{C}^2=2^N \bs{I}$, where $\bs I$ is the identity matrix, or, equivalently,
\begin{equation}\label{eq3:3}
    \bs{C}^{-1}=2^{-N}\bs{C}.
\end{equation}
\item 1-norm of $\bs C$ is
\begin{equation}\label{eq3:4}
    \|\bs C\|_1=\max_{0\leq k\leq N}\sum_{i=0}^N|c_{ik}|=2^N.
\end{equation}
\end{enumerate}
\end{proposition}
\begin{proof}
Consider the linear operator
$$
\mathcal Q_N\colon P(t)\longrightarrow(1-t^2)P'(t)-N(1-t)P(t),
$$
acting on the $(N+1)$-dimensional vector space of the polynomials $\R_N[t]$ of degree less or equal $N$. Direct calculations show that matrix $\bs Q_N$ is the matrix of $\mathcal Q_N$ in the standard basis $\{1,t,\ldots, t^N\}$ of $\R_N[t]$. Using the fact that for $P_k(t)$ defined by \eqref{eq3:2} we have that
$$
\mathcal Q_N\bigl(P_k(t)\bigr)=(1-t^2)P'_k(t)-N(1-t)P(t)=-2kP_k(t),\quad k=0,\ldots, N
$$
holds, we obtain that $P_k(t)$ can be considered as an eigenvector of $\mathcal Q_N$ corresponding to $q_k=-2k,\,k=0,\ldots,N$. Decomposing $P_k(t)$ through the standard basis, we obtain for some numbers $c_{ik}$ (see below \eqref{eq3:6})
$$
P_k(t)=\sum_{i=0}^Nc_{ik}t^i.
$$
Since $P_k(0)=c_{0k}=1$, which coincides with the condition on the columns of $\bs{C}$, assertions 1 and 2 of Proposition have been proved.

Note that $\sum_{k=0}^N c_{ik}c_{kj}$ is the coefficient at $t^i$ in the polynomial
\begin{align*}
\sum_{k=0}^N P_k(t)c_{kj}&=\sum_{k=0}^N c_{kj}(1-t)^k(1+t)^{N-k}=(1+t)^N\sum_{k=0}^Nc_{kj}\left(\frac{1-t}{1+t}\right)^k\\
&=(1+t)^NP_j\left(\frac{1-t}{1+t}\right)=(1+t)^N\left(1-\frac{1-t}{1+t}\right)^j\left(1+\frac{1-t}{1+t}\right)^{N-j}=2^Nt^{j}.
\end{align*}
This calculation implies that $\sum_{k=0}^Nc_{ik}c_{kj}=2^N$ if $i=j$, otherwise this expression is zero. This proves assertion 3. As a simple corollary note that
\begin{equation}\label{eq3:5}
    \sum_{k=0}^Nc_{k0}=2^N,\quad \sum_{k=0}^Nc_{kj}=0,\quad j\neq 0,
\end{equation}
since for the zeroth row of $\bs{C}$ we have $(c_{00},\ldots,c_{0N})=(1,\ldots,1)$.

Using \eqref{eq3:2} we find that
\begin{equation}\label{eq3:6}
    c_{ik}=\sum_{j=0}^N(-1)^j\binom{k}{j}\binom{N-k}{i-j}.
\end{equation}
Estimating the absolute value yields
$$
|c_{ik}|\leq \sum_{j=0}^N\binom{k}{j}\binom{N-k}{i-j}=\binom{N}{i}=c_{i0},
$$
which implies
$$
 \|\bs C\|_1=\max_{0\leq k\leq N}\sum_{i=0}^N|c_{ik}|=\sum_{i=0}^Nc_{i0}=\sum_{i=0}^{N}\binom{N}{i}=2^N.
$$
This finishes the proof.
\end{proof}

\begin{remark} Formulas \eqref{eq3:6} imply that for any $0\leq i,j\leq N$
\begin{equation}\label{eq3:6new}
    c_{ij}\binom{N}{j}=c_{ji}\binom{N}{i}.
\end{equation}

\end{remark}

\begin{remark}Recall that we consider only permutation invariant fitness landscapes. For the full fitness landscapes (i.e., not permutation invariant) and mutation matrices defined  by \eqref{eq0:7} it is more natural to consider a representation of the mutation matrix with the help of the Kronecker product (see, e.g., \cite{garcia2002linear} for the Crow--Kimura model). For our present purposes however we require the explicit form of the eigenvectors of $\bs Q$ as given by \eqref{eq3:6} and therefore do not use this type of representations.
\end{remark}

Proposition \ref{lem3:1} allows one to rewrite problem \eqref{eq1:5} in the basis composed of the eigenvectors of $\bs{Q}_N$. Let $\bs{D}_N=\diag(0,1,\ldots,N)$, then
\begin{equation}\label{eq3:7}
    \bs{C}^{-1}\bs{Q}_N\bs{C}=-2\bs{D}_N,\quad \bs{C}^{-1}\mu\bs{Q}_N\bs{C}=-2\mu\bs{D}_N.
\end{equation}
Multiplying \eqref{eq1:5} by $\bs{C}^{-1}$ from the left, we get
\begin{equation}\label{eq3:8}
\bs{C}^{-1}\bs{MCC}^{-1}\bs{p}+\mu\bs{C}^{-1}\bs{QCC}^{-1}\bs{p}=\overline{m}\bs{C}^{-1}\bs{p}.
\end{equation}
Introduce the notations
\begin{equation}\label{eq3:9}
\bs{F}(\bs{m})=\bs{C}^{-1}\bs{M}\bs{C}=2^{-N}\bs{CMC},\quad\bs{x}=\bs{C}^{-1}\bs{p}=2^{-N}\bs{Cp}.
\end{equation}
\eqref{eq3:9} and \eqref{eq3:8} together imply
\begin{equation}\label{eq3:10}
    \bs{F}(\bs m)\bs{x}=2\mu\bs{D}_N\bs{x}+\overline{m}\bs{x}.
\end{equation}
Matrix $\bs{F}(\bs m)=\bigl(f_{ik}(\bs m)\bigr)$ consists of linear forms with respect to $\bs{m}$, namely
$$
f_{ik}(\bs m)=2^{-N}\sum_{j=0}^Nm_jc_{ij}c_{jk}.
$$
The mean population fitness in new variables takes the form
\begin{equation}\label{eq3:11}
    \overline{m}=\bs{m}\cdot \bs{p}=\bs{C}^\top\bs{m}\cdot \bs{x}.
\end{equation}
The condition $\bs{p}\in S_{N+1}$ implies that
\begin{equation}\label{eq3:12}
    x_0=2^{-N},
\end{equation}
because the first row of $\bs{C}$ is composed of all ones, and due to \eqref{eq3:9}. In particular, image $\bs{C}^{-1}S_{N+1}$ of the simplex $S_{N+1}$ lies in the hyperplane defined by \eqref{eq3:12}.

Recall that if $\bs{v}^\top=(v_0,\ldots,v_N)\in\R^{N+1}$, then $\|\bs{v}\|_1=\sum_{k=0}^N|v_k|$. Therefore, for any $\bs{p}\in S_{N+1},$ $\|\bs{p}\|_1=1$. For any $\bs{x}\in \bs{C}^{-1}S_{N+1}$ by using Proposition \ref{lem3:1}
\begin{equation}\label{eq3:13}
    \|\bs{x}\|_1=2^{-N}\|\bs{Cp}\|_1\leq 2^{-N}\|\bs{C}\|_1\|\bs{p}\|_1=1.
\end{equation}
Moreover, using \eqref{eq3:9}, \eqref{eq1:5} and \eqref{eq3:13}:
\begin{equation}\label{eq3:14}
\begin{split}
  \|\bs{F}(\bs m)\bs{x}\|_1&=2^{-N}\|\bs{CMp}\|_1 \leq 2^{-N}\|\bs{C}\|_1\|\bs{M}\|_1\|\bs{p}\|_1=\|\bs{M}\|_1=\max_{0\leq k\leq N} \{m_k\}.
\end{split}
\end{equation}
Additionally,
\begin{equation}\label{eq3:15}
\begin{split}
  \|\overline{m}\bs{x}\|_1&=\left|\sum_{k=0}^N m_kp_k\right|\|\bs{x}\|_1\leq \max_{0\leq k\leq N} \{m_k\} \|\bs{p}\|_1=\|\bs{M}\|_1.
\end{split}
\end{equation}

\section{Asymptotic stabilization of $\bs{p}(\mu)$ when $\mu\to\infty$}
In this section we study the behavior of $\bs{p}(\mu)$ for large $\mu$. One important case from the evolutionary viewpoint is when the distribution of the types of sequences becomes uniform. This implies that natural selection ceases to operate in the population.  In terms of the classes of sequences, for which our model is written and which constitute our vector $\bs{p}(\mu)$, this means that the distribution of the classes is binomial (recall that class $k$ has ${N \choose k}$ types of sequences).

First we show that the binomial distribution can be a solution of \eqref{eq1:5} only if the fitness matrix is a scalar matrix, i.e., the identity matrix times some constant.

\begin{proposition}\label{lem4:1} If $\bs{M}$ is different from $m\bs{I}$, where $m$ is a constant, then $\bs{p}'(\mu)\neq 0$ for any $\mu\geq 0$.
\end{proposition}
\begin{proof}
Assume that $\bs{M}\neq m\bs{I}$ and for some $\mu_0\geq 0$ we have that $\bs{p}'(\mu_0)=0$. By differentiating \eqref{eq1:5}, we get
\begin{equation}\label{eq4:1}
    \bs{M}\bs{p}'+\bs{Qp}+\mu \bs{Qp}'=\overline{m}'\bs{p}+\overline{m}\bs{p}'.
\end{equation}
Putting $\mu=\mu_0$ into \eqref{eq4:1} and using $\overline{m}'(\mu_0)=\sum_{k=0}^N m_kp_k'(\mu_0)=0$ imply
\begin{equation}\label{eq4:2}
\bs{Q}\bs{p}(\mu_0)=0,
\end{equation}
i.e., $\bs{p}(\mu_0)$ is an eigenvector of $\bs{Q}$ corresponding to the zero eigenvalue. Due to Proposition \ref{lem3:1} (see \eqref{eq3:6}) this eigenvector is given by
$$
\bs{p}(\mu_0)=2^{-N}\left(\binom{N}{0},\binom{N}{1},\ldots,\binom{N}{N}\right)^\top,
$$
i.e., this is the vector of the binomial distribution with $p=1/2$.

Plugging \eqref{eq4:2} into \eqref{eq1:5} yields $\bs{M}\bs{p}(\mu_0)=\overline{m}\bs{p}(\mu_0),$ or, in coordinates
$$
m_k\binom{N}{k}=\overline{m}\binom{N}{k},\quad k=0,\ldots,N,
$$
which implies $m_0=\ldots=m_N=\overline{m}$, which contradicts the initial assumption.
\end{proof}

Proposition \ref{lem4:1} shows that for any non scalar matrix $\bs{M}$ the vector $\bs{p}$ is not binomial.

However, if $\mu\to\infty$ then $\bs{p}(\mu)$ approaches the binomial distribution, which can be easily shown. To illustrate the utility of the change of the basis from the previous section, we provide here estimates for the speed of convergence, given in terms of the parameters of the model. To formulate our result, we shall use the following definition.
\begin{definition}\label{def4:1} We shall say that the solution $\bs{x}(\mu)=\bs{C}^{-1}\bs{p}(\mu)$ to the problem \eqref{eq3:10} admits asymptotic stabilization if there exists a vector $\bs{\hat{x}}=\bs{C}^{-1}\bs{\hat{p}}\in \bs{C}^{-1}S_{N+1}$ such that
\begin{equation}\label{eq4:3}
    \lim_{\mu\to\infty}\bs{x}(\mu)=\bs{\hat{x}},\quad \lim_{\mu\to\infty}\bs{x}'(\mu)=0.
\end{equation}
\end{definition}
Definition \ref{def4:1} is equivalent to the following: for any $\varepsilon>0$ there exists $\mu_{\varepsilon}$ such that for any $\mu>\mu_{\varepsilon}$ the following inequalities hold:
\begin{equation}\label{eq4:4}
    \|\bs{x}(\mu)-\bs{\hat{x}}\|_1<\varepsilon,\quad \|\bs{x}'(\mu)\|_1<\varepsilon.
\end{equation}
Definition \ref{def4:1} can be also given in terms on the solution to \eqref{eq1:5}: vector $\bs{p}(\mu)$ admits the asymptotic stabilization if
\begin{equation}\label{eq4:5}
    \lim_{\mu\to\infty}\bs{p}(\mu)=\bs{\hat{p}},\quad \lim_{\mu\to\infty}\bs{p}'(\mu)=0.
\end{equation}
Moreover, in this case, using \eqref{eq1:4},
\begin{equation}\label{eq4:6}
    \lim_{\mu\to\infty}\overline{m}(\mu)=\hat{\overline{{m}}},\quad \lim_{\mu\to\infty}\overline{m}'(\mu)=0.
\end{equation}

\begin{theorem}\label{th4:1}
For any matrix $\bs{M}=\diag(m_0,\ldots,m_N)$ the solution $\bs x$ to \eqref{eq3:10} satisfies
\begin{equation*}
    \|\bs{x}-\bs{\hat{x}}\|_1\leq \frac{1}{\mu}\|\bs{M}\|_1,\quad
    \|\bs{x}'\|_1 \leq \frac{2N}{2\mu-(2^{N+1}+1)\|\bs{M}\|_1}\,,
\end{equation*}
where
\begin{equation}\label{eq4:7}
    \bs{\hat{x}}^\top=2^{-N}(1,0,\ldots,0).
\end{equation}
This implies that $\bs x$ admits the asymptotic stabilization.
\end{theorem}
\begin{proof} Since $\bs{\hat{x}}$ is an eigenvector of $\bs{D}_N$ corresponding to the zero eigenvalue, then $\bs{D}_N\bs{\hat{x}}=0$; therefore \eqref{eq3:10} yields
\begin{equation}\label{eq4:8}
    \bs{F}(\bs m)\bs{x}-\overline{m}\,\bs{x}=2\mu\bs{D}_N(\bs{x}-\bs{\hat{x}}).
\end{equation}
Taking into account \eqref{eq3:14} and \eqref{eq3:15}, the left hand side of \eqref{eq4:8} can be estimated
\begin{equation}\label{eq4:9}
    \| \bs{F}(\bs m)\bs{x}-\overline{m}\,\bs{x}\|_1\leq 2\|\bs{M}\|_1=2\max_{0\leq k\leq N} m_k.
\end{equation}
On the other hand
\begin{equation}\label{eq4:10}
    \|2\mu\bs{D}_N(\bs{x}-\bs{\hat{x}})\|_1=2\mu\sum_{k=1}^Nk|x_k-\hat{x}_k|\geq 2\mu \|\bs{x}-\bs{\hat{x}}\|_1.
\end{equation}
From \eqref{eq4:8}--\eqref{eq4:10}, assuming $\mu>0$, we find $2\mu \|\bs{x}-\bs{\hat{x}}\|_1\leq 2\|\bs{M}\|_1$, i.e.,
\begin{equation}\label{eq4:11}
    \|\bs{x}-\bs{\hat{x}}\|_1\leq \frac{1}{\mu}\|\bs{M}\|_1,
\end{equation}
from which it follows that $\bs{x}(\mu)\to\bs{\hat{x}}$ when $\mu\to\infty$.

Differentiating \eqref{eq3:10} with respect to $\mu$, we obtain
\begin{equation}\label{eq4:12}
    \bs{F}(\bs m)\bs{x}'-\overline{m}'\bs{x}-\overline{m}\,\bs{x}'-2\bs{D}_N\bs{x}=2\mu\bs{D}_N\bs{x}'.
\end{equation}
The left hand side can be bounded as
\begin{equation}\label{eq4:13}
    \|\bs{F}(\bs m)\bs{x}'-\overline{m}'\bs{x}-\overline{m}\,\bs{x}'-2\bs{D}_N\bs{x}\|_1\leq (2^{N+1}+1)\|\bs{M}\|_1\|\bs{x}'\|_1+2N.
\end{equation}
The estimate \eqref{eq4:13} follows from the expression \eqref{eq3:9} and estimates \eqref{eq3:14}, \eqref{eq3:15}:
\begin{align*}
\|\bs{F}(\bs{m})\bs{x}'\|_1&\leq \|\bs{F}(\bs{m})\|_1\|\bs{x}'\|_1=2^{-N}\|\bs{CMC}\|_1 \|\bs{x}'\|_1
\leq 2^{N}\|\bs{M}\|_1\|\bs{x}'\|_1;\\
\|\overline{m}'\bs{x}\|_1&\leq |\overline{m}'|=\left|\sum_{k=0}^Nm_kp_k\right|\leq \|\bs{M}\|_1\|\bs{p}'\|_1=\|\bs{M}\|_1\|\bs{Cx}'\|_1
\leq 2^N \|\bs{M}\|_1\|\bs{x}'\|_1;\\
\|\overline{m}\bs{x}'\|_1&\leq |\overline{m}|\|\bs{x}'\|_1=\left|\sum_{k=0}^N m_kp_k\right|\|\bs{x}'\|_1
\leq \|\bs{M}\|_1\|\bs{x}'\|_1;\\
\|2\bs{D}_N\bs{x}\|_1&\leq 2 \|\bs{D}_N\|_1\|\bs{x}\|_1\leq 2N.
\end{align*}
For the right hand side of \eqref{eq4:12}, utilizing $x_0\equiv 2^{-N}$ and $x_0'\equiv 0$, we obtain a lower bound as
\begin{equation}\label{eq4:14}
    \|2\mu\bs{D}_N\bs{x}'\|_1=2\mu\sum_{k=1}^N k|x_k'|\geq 2\mu\sum_{k=1}^N|x_k'|=2\mu \|\bs{x}'\|_1.
\end{equation}
Putting together \eqref{eq4:12}--\eqref{eq4:14}, we find $2\mu \|\bs{x}'\|_1\leq (2^{N+1}+1)\|\bs{M}\|_1\|\bs{x}'\|_1+2N$. Therefore, for large enough $\mu$
\begin{equation}\label{eq4:15}
    \|\bs{x}'\|_1 \leq \frac{2N}{2\mu-(2^{N+1}+1)\|\bs{M}\|_1}\,,
\end{equation}
which implies \eqref{eq4:3}.
\end{proof}
\begin{remark}
Returning back to vector $\bs{p}$, we find that
$$
\bs{p}(\mu)\to\bs{\hat{p}}=\bs{C\hat{x}}=2^{-N}\left(\binom{N}{0},\binom{N}{1},\ldots,\binom{N}{i},\ldots,\binom{N}{N}\right)^\top,
$$
and
$$
\overline{m}(\mu)\to\hat{\overline{m}}=2^{-N}\sum_{k=0}^N\binom{N}{k}m_k,\quad \overline{m}'(\mu)\to 0.
$$

\end{remark}

\section{Parametric solution to equation \eqref{eq1:5}}\label{sec5} Without loss of generality matrix $\bs{M}$ can be taken such that $\bs{M}\neq m\bs{I}$ (in this case the solution is given by the vector of binomial distribution, see Proposition \ref{lem4:1}), and also $\min_{k} \{m_k\}=0$. The original equation $\bs{Mp}+\mu\bs{Q}_N\bs{p}=\overline{m}\bs{p}$ can be rewritten as
\begin{equation}\label{eq5:1}
    \bs{Mp}=\overline{m}\bs{p}-\mu \bs{Q}_N\bs{p}.
\end{equation}
As before, we use the substitution $\bs{p}=\bs{Cx}$:
\begin{equation}\label{eq5:2}
    \bs{Mp}=\bs{C}\bigl(\overline{m}\bs{I}+2\mu\diag(0,1,\ldots,N)\bigr)\bs{x},
\end{equation}
or, using the parameter $s=2\mu/\overline{m}$,
\begin{equation}\label{eq5:3}
    \bs{Mp}=\overline{m}\bs{C}\bigl(\bs{I}+\diag(0,s,2s,\ldots,Ns)\bigr)\bs{x},\quad \bs{p}=\bs{Cx}.
\end{equation}
In coordinates, we have
\begin{equation}\label{eq5:4}
  m_ip_i =\overline{m}\sum_{k=0}^N c_{ik}(1+ks)x_k,\quad
    p_i=\sum_{k=0}^N c_{ik}x_k,\quad i=0,\ldots,N.
\end{equation}
We are looking for a parametric solution $\bs{p}=\bs{p}(s),\,\bs{x}=\bs{x}(s)$. Introduce the following notations:
\begin{equation}\label{eq5:5}
    F(s)=\overline{m}(s),\quad \mu=\frac{1}{2}sF(s).
\end{equation}
We can try to express the vector $\bs{x}$ through the function $F(s)$, which is actually the dominant eigenvalue, and through $\bs{x}$ we can find $\bs{p}=\bs{Cx}$.

Recall that for any $s$
\begin{equation}\label{eq5:6}
    \sum_{k=0}^Np_k(s)=1,\quad x_0(s)=2^{-N}.
\end{equation}

In the following we deal only with the simplest case when all the fitnesses except one are zero. For the general case we refer to Appendix \ref{append:B}.

Let $m_j>0$ for some $j$ and all other $m_i=0,i=0,\ldots,N,i\neq j$. Then $F(s)=m_jp_j(s)>0$ at least for small $s$ since $F(s)\to m_j$ as $s\to 0$. System \eqref{eq5:4} takes the form
\begin{equation}\label{eq5:7}
\begin{split}
  m_jp_j &=\overline{m}\sum_{k=0}^N c_{jk}(1+ks)x_k,\\
  0&=\overline{m}\sum_{k=0}^N c_{ik}(1+ks)x_k,\quad i=0,\ldots,N,\,i\neq j,\\
    p_i& =\sum_{k=0}^N c_{ik}x_k,\quad i=0,\ldots,N.
\end{split}
\end{equation}
After dividing by $\overline{m}$, we have
\begin{equation}\label{eq5:8}
\begin{split}
  1 &=\sum_{k=0}^N c_{jk}(1+ks)x_k,\\
  0&=\sum_{k=0}^N c_{ik}(1+ks)x_k,\quad i=0,\ldots,N,\,i\neq j,\\
    p_i& =\sum_{k=0}^N c_{ik}x_k,\quad i=0,\ldots,N.
\end{split}
\end{equation}
In the matrix form system \eqref{eq5:8} is
\begin{equation}\label{eq5:9}
    \bs{e}_j=\bs{C}\diag(1,1+s,1+2s,\ldots,1+Ns)\bs x,\quad \bs{p}=\bs{Cx},
\end{equation}
where $\bs{e}_j$ is the $j$-th standard unit vector. From \eqref{eq5:8} and \eqref{eq3:3}:
$$
\diag(1,1+s,\ldots,1+Ns)\bs{x}=2^{-N}\bs{Ce}_j,
$$
or, in coordinates,
\begin{equation}\label{eq5:10}
    x_k(s)=2^{-N}\frac{c_{kj}}{1+ks}\,,\quad p_i(s)=2^{-N}\sum_{k=0}^N\frac{c_{ik}c_{kj}}{1+ks}\,,\quad \overline{m}(s)=m_jp_j(s),\quad \mu=\frac{1}{2}s\overline{m}(s).
\end{equation}
If one uses the notation
\begin{equation}\label{eq5:11}
    F_{ij}(s)=2^{-N}\sum_{k=0}^N\frac{c_{ik}c_{kj}}{1+ks}\,,
\end{equation}
then the solution \eqref{eq5:10} can be represented in the following compact form:
\begin{equation}\label{eq5:12}
    p_i(s)=F_{ij}(s),\quad F(s)=m_jF_{jj}(s),\quad \mu=\frac{s}{2}F(s),\quad \overline{m}=F(s).
\end{equation}
\begin{remark} Note that if $j=0$, i.e., the fittest type corresponds to the zeroth index, the solutions \eqref{eq5:10} or \eqref{eq5:12} can be simplified further. We use here that $c_{i0}=\binom{N}{i}$ and $c_{0k}=1$. This implies that
\begin{equation}\label{eq5:13}
\begin{split}
    x_k(s)&=2^{-N}\frac{\binom{N}{k}}{1+ks}\,,\quad p_i(s)=2^{-N}\sum_{k=0}^N\frac{c_{ik}\binom{N}{k}}{1+ks}\,,\\
     \overline{m}(s)&=m_0p_0(s)=\frac{m_0}{2^{N}}\sum_{k=0}^N\frac{\binom{N}{k}}{1+ks}\,,\quad \mu=\frac{1}{2}s\overline{m}(s).
\end{split}
\end{equation}
\end{remark}
A significant number of results about the quasispecies theory is formulated in terms of some limit procedures, when $N\to\infty$. A care should be exercised in this case, since different scalings are possible for the fitness landscape $\bs M$ and the mutation rates. The exact parametric solution \eqref{eq5:10} obtained in this section (see also Appendix \ref{append:B} for a general approach) can be profitably used to obtain such asymptotic expressions, see the next section for two examples.
\section{Two examples of particular fitness landscapes}
In this section we give two specific examples of the applications of the obtained parametric solutions \eqref{eq5:10}. First, we show that our exact methods work in a specific case, which is not covered by the existing in the literature approaches. Second, we obtain an exact solution, for both the leading eigenvalue and the corresponding eigenvector, in the case of the single peaked landscape for $N\to\infty$ and also provide estimates of the speed of convergence. An ``exact solution of the quasispecies model'' exits in the literature \cite{galluccio1997exact}, see also \cite{saakian2004ese} for an ad hoc approach to the same problem, however, we present a mathematically rigorous derivation of the limit quasispecies distribution (see also the discussion in Section \ref{example:2}).

\subsection{Example 1: $\bs m=(0,\ldots,0,N,0,\ldots,0)$}\label{example:1}
To show how the parametric solution from Section \ref{sec5} works, consider the following example. Let the sequence length be even, i.e., $N=2A$ for some integer $A$, and let the fitness landscape be defined as
\begin{equation}\label{eqAA:3}
\bs m=(0,\ldots,0,N,0,\ldots,0),
\end{equation}
where $N$ is exactly at the $A$-th place. Together with $\bs m$ consider also scaled fitness landscape $N\bs r=\bs m$, and the corresponding mean fitnesses $\overline{m}(\mu)$ and $\overline{r}(\mu)=\overline{m}(\mu)/N$. From \eqref{eq5:10} we have
$$
\overline{r}(s)=\frac{1}{2^{2A}}\sum_{i=0}^{2A}\frac{c_{Ak}c_{kA}}{1+ks}\,,\quad \mu(s)=As\overline r(s).
$$
Using \eqref{eq3:2} and \eqref{eq3:6new}, we find
$$
c_{Ak}c_{kA}=\frac{\binom{2A}{A}}{\binom{2A}{k}}c_{kA}^2=\begin{cases}
0,&k=2l+1,\\
\binom{2l}{l}\binom{2(A-l)}{A-l},&k=2l.
\end{cases}
$$
Then \eqref{eq5:10} takes the form
\begin{equation}\label{eqAA:1}
\overline{r}(s)=\frac{1}{2^{2A}}\sum_{l=0}^{A}\frac{\binom{2l}{l}\binom{2(A-l)}{A-l}}{1+2ls}\,,\quad \mu(s)=As\overline r(s).
\end{equation}
Using the fact that
$$
\frac{1}{2^{2n}}\binom{2n}{n}\approx \frac{1}{\sqrt{\pi n}}\,,
$$
we find that in \eqref{eqAA:1} for $0<l<A$
$$
\frac{1}{2^{2A}}\binom{2l}{l}\binom{2(A-l)}{A-l}=\frac{1}{2^{2l}}\binom{2l}{l}\frac{1}{2^{2(A-l)}}\binom{2(A-l)}{A-l}\approx \frac{1}{\pi\sqrt{l(A-l)}}\,.
$$
Now fix $\mu$ and assume that $\overline{r}\to\overline{r}_\infty$ for $A\to\infty$. Then for $A\gg 1$
$$
s=\frac{\mu}{A\overline{r}}\approx \frac{\mu}{A\overline{r}_\infty}\,.
$$
For \eqref{eqAA:1} we find
$$
\overline{r}_\infty\approx \frac{1}{2^{2A}}\binom{2A}{A}+\frac{1}{\pi}\sum_{l=1}^{A-1}\frac{1}{\sqrt{l(A-l)}(1+\frac{2l\mu}{A\overline{r}_\infty})}+\frac{1}{2^{2A}}\binom{2A}{A}\frac{1}{1+\frac{2\mu}{\overline{r}_\infty}}\,.
$$
The first and the last terms tend to zero as $A\to\infty$, and the middle term is
$$
\lim_{A\to\infty}\sum_{l=1}^{A-1}\frac{1}{\sqrt{l(A-l)}(1+\frac{2l\mu}{A\overline{r}_\infty})}=\int_{0}^1\frac{dx}{\sqrt{x(1-x)}(1+\frac{2x\mu}{\overline{r}_\infty})}\,,
$$
where the last intergal can be evaluated exactly. Finally, we find that $\overline{r}_\infty$ is determined from
$$
\overline{r}_\infty=\frac{1}{\sqrt{1+\frac{2\mu}{\overline{r}_\infty}}}\,,
$$
which yields
\begin{proposition}\label{propAA:1} For the Crow--Kimura quasispecies model with the fitness landscape \eqref{eqAA:3}
\begin{equation}\label{eqAA:2}
    \overline{r}_\infty(\mu)=\lim_{A\to\infty} \overline r(\mu)=\sqrt{\mu^2+1}-\mu.
\end{equation}
\end{proposition}
A comparison of the obtained formula with the numerical computations is given in Fig. \ref{figA:1}.
\begin{figure}[!th]
\centering
\includegraphics[width=\textwidth]{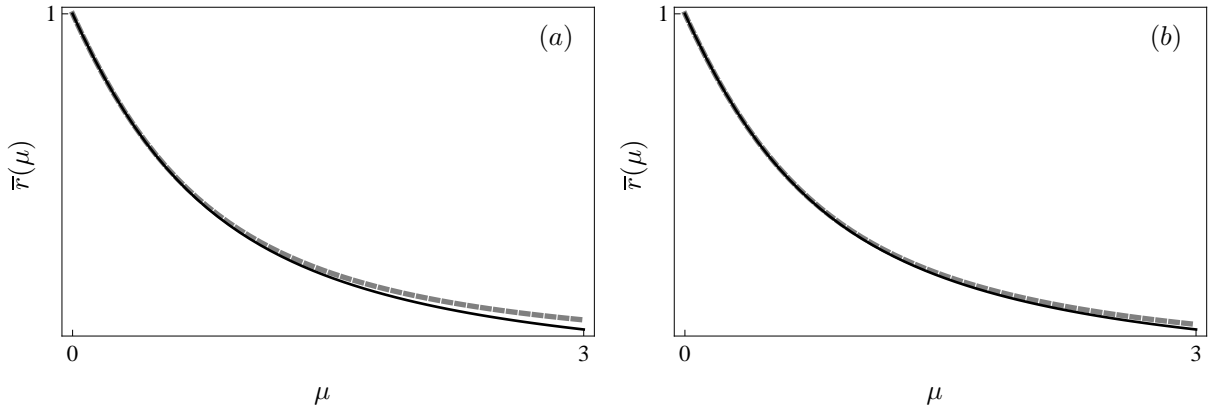}
\caption{Comparison of the limit \eqref{eqAA:2} (black) with the numerical computations (gray, dashed) for the fitness landscape \eqref{eqAA:3}. $(a)$ The sequence length is $N=100$, $(b)$ $N=200$}\label{figA:1}
\end{figure}
\begin{remark}In addition to the exact result \eqref{eqAA:2}, we have checked numerically that for $N=2A\gg 1$ and $k\in \N$, we have
\begin{equation}\label{eqAA:12}
p_A(\mu)\approx \overline{r}_\infty=\sqrt{\mu^2+1}-\mu,\quad p_{A\pm k}\approx \overline{r}_\infty \left(\frac{1-\overline{r}_\infty }{1+\overline{r}_\infty }\right)^k,
\end{equation}
that is, presumably, the limit distribution is two sided geometric (see Fig. \ref{figA:2}) (the cases $p_A,\,p_{A\pm 1}$ can be straightforwardly proved by using the parametric solution \eqref{eq5:10}).
\begin{figure}[!th]
\centering
\includegraphics[width=\textwidth]{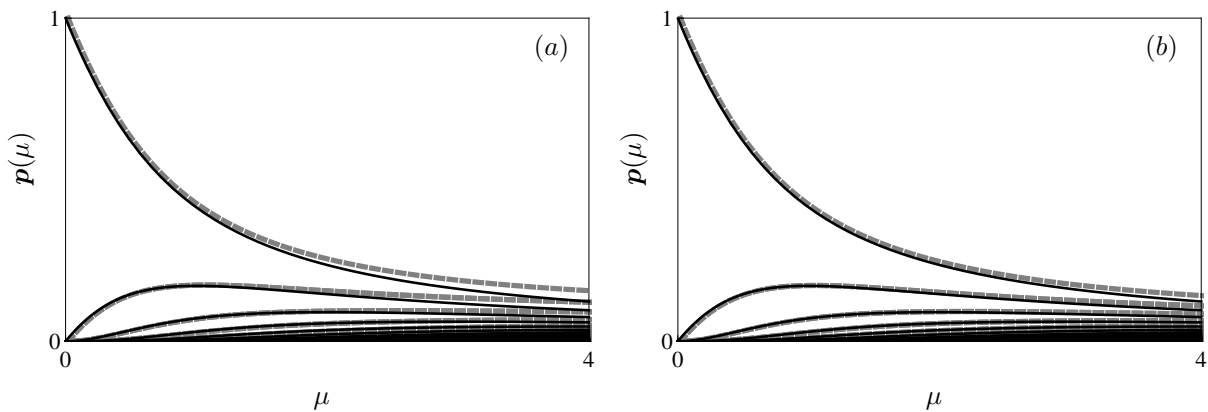}
\caption{Comparison of the limit distribution \eqref{eqAA:12} (black) with the numerical computations (gray, dashed) for $N=100$ $(a)$ and $N=200$ $(b)$}\label{figA:2}
\end{figure}
\end{remark}

In \cite{Baake2007,Hermisson2002} a maximum principle for the quasispecies model was formulated, which, using the notations of the present text, can be stated as follows (we formulate it in the form convenient for comparison with our results and note that more general situations are also treated in the cited papers). Assume that $m_i=N r_i=N r(x_i),\,x_i=\frac{i}{N}\in[0,1]$ and define $g(x)=\mu \bigl(1-2\sqrt{x(1-x)}\bigr)$. Then the scaled mean fitness $\overline{r}(\mu)=\overline{m}(\mu)/N$ is given by
\begin{equation}\label{eqAA:10}
    \overline{r}\approx \overline{r}_\infty=\sup_{x\in[0,1]}\bigl(r(x)-g(x)\bigr).
\end{equation}
The maximum principle \eqref{eqAA:10} holds when some additional technical conditions are satisfied. In particular, in \cite{Hermisson2002} it was assumed that $r(x)$ may have only finite number of discontinuities and be either left or right continuous at every point, which obviously does not hold for the fitness landscape \eqref{eqAA:3}. Formal application of the maximal principle \eqref{eqAA:10} to \eqref{eqAA:3} leads to incorrect results (i.e., for $\mu=1$ it predicts that $\overline{r}\approx 1$, which is wrong, see Fig.~\ref{figA:1}). The last conclusion emphasizes the importance of careful limit procedures and value of the exact formulas, which can be used on a case by case basis.
\subsection{Exact solution of the model with the single peaked fitness landscape}\label{example:2}
Consider the fitness landscape
$$
\bs m=N\bs r=N(1,0,\ldots,0).
$$
We use here the same notation as in Section \ref{example:1}, i.e., $\overline{r}(\mu)=\overline{m}(\mu)/N$. We have from \eqref{eq5:10} that
\begin{equation}\label{eqQ:1}
    \overline{r}(s)=\frac{1}{2^N}\sum_{k=0}^N\frac{\binom{N}{k}}{1+ks}\,,\quad \mu(s)=\frac{Ns}{2}\overline{r}(s),
\end{equation}
and the components of the leading eigenvector can be found as
$$
p_a=p_a(s)=\frac{1}{2^N}\sum_{k=0}^N \frac{\binom{N}{k}c_{ak}}{1+ks},\quad a=0,\ldots, N,
$$
or, using the previous,
\begin{equation}\label{eqQ:2}
    p_0=\overline{r}=\frac{1}{2^N}\sum_{k=0}^N\frac{\binom{N}{k}}{1+\frac{2\mu}{\overline{r}}\frac{k}{N}}\,,\quad  p_a=\frac{\binom{N}{a}}{2^N}\sum_{k=0}^N\frac{c_{ka}}{1+\frac{2\mu}{\overline{r}}\frac{k}{N}}\,,\quad a=1,\ldots,N.
\end{equation}

The rest of this section is devoted to the proof of 
\begin{proposition}\label{propQ:1}
For fixed $a$ and $\mu<1$ the limit distribution for $N\to\infty$ for the Crow--Kimura quasispecies model with the single peaked landscape is geometric:
\begin{equation}\label{eqQ:6}
    \overline{r}_{\infty}=\lim\limits_{N\to\infty}p_0=1-\mu\,,\quad
\lim\limits_{N\to\infty}p_a=(1-\mu)\mu^a\,,\quad a\geq
1\,.
\end{equation}
If $\mu\geq 1$ then $\overline{r}_\infty=0$ and the quasispecies distribution is degenerate.
\end{proposition}
\begin{proof}Introduce the notation $u=\dfrac{\mu}{\overline{r}}\,,$ and note that due to \eqref{eq3:2} we have
$$
\sum_{k=0}^N{N\choose k}=P_0(1)=2^N,\quad \sum_{k=0}^Nc_{ka}=P_a(1)=0,\quad a\geq 1.
$$
Therefore, for $a=0$ we obtain
$$
\overline{r}-\frac{1}{1+u}=\frac{1}{2^{N}}\sum_{k=0}^{N} {N\choose
k}\left(\frac{1}{1+2u\frac{k}{N}}-\frac{1}{1+u}\right)=\frac{u}{N
2^{N}}\sum_{k=0}^{N} {N\choose
k}\frac{N-2k}{\left(1+2u\frac{k}{N}\right)(1+u)}\,.
$$
Using the facts that ${N\choose k}={N\choose N-k}$ and $N-2(N-k)=-(N-2k)$, we get
\begin{align*}
\overline{r}-\frac{1}{1+u}&=\frac{u}{N^{N}}\sum_{2k\leq N}
{N\choose k}(N-2k)\left(\frac{1}{\left(1+2u\frac{k}{N}\right)(1+u)}-
\frac{1}{\left(1+2u\frac{N-k}{N}\right)(1+u)}\right)\\
&=\frac{2u^2}{(1+u)N^22^{N}}\sum_{2k\leq N}
\frac{{N\choose k}(N-2k)^2}{\left(1+2u\frac{k}{N}\right)\left(1+2u\frac{N-k}{N}\right)}\\
&\leq \frac{u^2}{(1+u)N^2 2^{N}}\sum_{k=0}^{N} {N\choose
k}(N-2k)^2\,.
\end{align*}
Since $2^{-N}\sum_{k=0}^N{N\choose k}(N-2k)^2=N$ (e.g., here the left hand side is 4 times the variance of the binomial random variable with parameters $N$ and $1/2$, which is $N/4$), the last inequality yields the estimate
\begin{equation*}
    0\leq \overline{r}-\frac{1}{1+u}\leq \frac{u^2}{(1+u)N}\,,
\end{equation*}
or, returning to the original parameters,
\begin{equation}\label{eqQ:9}
    0\leq \overline{r}^2(\overline{r}+\mu-1)\leq \frac{\mu^2}{N}\,.
\end{equation}
The inequality \eqref{eqQ:9} implies
\begin{itemize}
\item the estimate $\overline{r}\geq 1-\mu$;
\item for $\mu=1$ the estimate $\overline{r}^3\leq 1/N$, which means that $\lim_{N\to\infty}\overline{r}=0$ for $\mu=1$;
\item for the case $\mu>1$ the estimate $\overline{r}^3\leq \overline{r}^2(\overline{r}+\mu-1)\leq \mu^2/N$ and hence $\lim_{N\to\infty}\overline{r}=0$ for $\mu>1$. Therefore, in the limit $N\to\infty$ the distribution $\bs p$ is degenerate;
\item in the case $\mu<1$ the estimate
$$0\leq \overline{r}+\mu-1\leq\frac{\mu^2}{\overline{r}^2N}\leq \frac{\mu^2}{(1-\mu)^2N}\,,$$
which proves the first equality in \eqref{eqQ:6}, and also gives an estimate of the speed of convergence.
\end{itemize}
\begin{remark}At this point we would like to note that the first equality in \eqref{eqQ:6} is a well known fact, which originally was proved in \cite{galluccio1997exact} and also elementary follows from the maximum principle \eqref{eqAA:10}. We present a full proof of this fact to illustrate the general approach by the parametric solution \eqref{eq5:10} and specify an exact constant in the expression $\overline{r}=\overline{r}_\infty+O(N^{-1})$, which is, to the best of our knowledge, new. In \cite{galluccio1997exact} also the expression for $p_0$ is given, which is basically $\overline{r}_\infty$, and an integral representation of the components of the quasispecies vector (see (49) in the cited text); to actually compute this distribution or obtain approximations with easily obtained error estimates are a separate nontrivial problem, as can be seen in \cite{galluccio1997exact}. In \cite{saakian2004ese} the limit distribution for the quasispecies was derived without proof of convergence. The rest of the our proof shows that actually the limit distribution is \textit{geometric} and gives estimates of the speed of convergence.
\end{remark}

Now we treat the {case} $a\geq 1$.

It can be shown (see the proof of Lemma \ref{lemQ:1} below) that for fixed $N$
\begin{equation}\label{eqQ:14}
    p_a=\frac{{N\choose a}\cdot a!  u^a}{N^a
2^{N-a}}\sum_{k=0}^{N-a} \frac{{N-a\choose
k}}{\prod\limits_{j=k}^{k+a}\left(1+2u\frac{j}{N}\right)}\,.
\end{equation}
We know that for $\mu\geq 1$ the distribution $\bs p$ is degenerate, therefore we are only interested in the case $\mu<1$, therefore $u=\mu/\overline{r}\to\mu/(1-\mu)$, and hence to conclude the proof we need to show that
for $u\geq 0$
\begin{equation}\label{eqQ:15}
   \lim_{N\to\infty} p_a= \lim\limits_{N\to\infty} \frac{{N\choose a} a! u^a}{N^a
2^{N-a}}\sum_{k=0}^{N-a} \frac{{N-a\choose
k}}{\prod\limits_{j=k}^{k+a}\left(1+2u\frac{j}{N}\right)}=\frac{u^a}{(1+u)^{a+1}}\,.
\end{equation}
\begin{lemma}\label{lemQ:1}
The estimate
\begin{equation}\label{eqQ:27}
    \left|p_a-\frac{u^a}{(1+u)^{a+1}}\right|\leq\frac{u^{a+1}
2^a(a+1)}{\sqrt{N}}+\left(1-\frac{{N\choose a} a!
}{N^a}\right)\frac{u^a}{(1+u)^{a+1}}\,
\end{equation}
holds.
\end{lemma}
We defer the proof of Lemma \ref{lemQ:1} to Appendix \ref{append:D} and note that the estimate of the speed of convergence can be simplified by taking into account that
$$
0\leq 1-\frac{{N\choose a}a!}{N^a}\leq \frac{a(a-1)}{2N}\,,
$$
which can be proved by, e.g., induction.

Lemma \ref{lemQ:1} implies the limit \eqref{eqQ:15}, and therefore
$$
\lim_{N\to\infty}p_a=\lim_{N\to\infty}\frac{u^a}{(1+u)^{a+1}}=\mu^a(1-\mu),
$$
which concludes the proof of Proposition \ref{propQ:1}.
\end{proof}
In Figure \ref{figQ:1} a comparison of numerical solutions of the quasispecies model with the single peaked fitness landscape for $N=100$ with the geometric limit distribution \eqref{eqQ:6} is given.
\begin{figure}[!th]
\centering
\includegraphics[width=\textwidth]{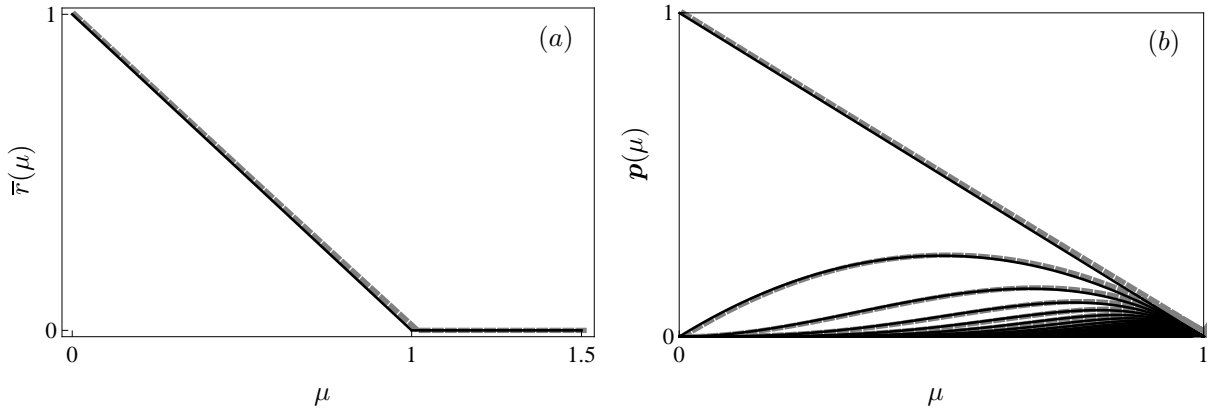}
\caption{Comparison of numerical solutions of the quasispecies model (gray, dashed) with the single peaked fitness landscape for $N=100$ with the geometric limit distribution \eqref{eqQ:6} (black). $(a)$ The expressions for the leading eigenvalue; $(b)$ the quasispecies distribution}\label{figQ:1}
\end{figure}

\section{Approximate formulas for epsilon stabilization}\label{ep_stab}
Due to Theorem \ref{th4:1} we know that for any problem of the form \eqref{eq1:5} the solution admits asymptotic stabilization, which immediately implies that the dominant eigenvalue approaches $\hat{\overline{m}}$ as $\mu\to\infty$. However, from the practical point of view the following question can be asked: Given a small fixed $\varepsilon$ can we find the critical value of the mutation rate $\mu^*_\varepsilon$ for which $\overline{m}(\mu)$ is $\varepsilon$-close to the limit value? In this section we suggest an approximate heuristic solution to this question which, in some sense, generalizes the textbook formula for the classical error threshold in the quasispecies theory (see the next section).

In the following we will need several facts about perturbations of the eigenvalues of a linear operator, which we collect here. First, function $\overline{m}(\mu)$ is smooth, which follows from the results on the perturbation of a simple eigenvalue of a matrix \cite{kato1995perturbation,rellich1969perturbation,vishik1960solution}. In \cite{kato1995perturbation} explicit formulas for the derivatives of $\overline{m}(\mu)$ are given, which in our notations can be written as follows. Let $\mu=\mu_0\geq 0,\,\overline{m}(\mu_0)=\overline{m}_0,\,\bs{p}(\mu_0)=\bs{p}^0$. Then
\begin{equation}\label{eq2:1}
    \left.\frac{d\overline{m}}{d\mu}\right|_{\mu=\mu_0}=\overline{m}'(\mu_0)=\bs{Qp}^0\cdot \bs q^0,
\end{equation}
where $\bs{q}^{0}$ is the eigenvector of the adjoint problem
$
(\bs{M}+\mu_0\bs{Q}^\top)\bs{q}^{0}=\overline{m}_0\bs{q}^{0},
$ normalized such that $\bs{p}^0\cdot \bs{q}^{0}=1$ holds (eigenvector $\bs q^0$ is actually related to the so-called \emph{ancestral} distribution, see \cite{Hermisson2002} for more details). For the case $m_0>m_1\geq \ldots\geq m_N$ and $\mu_0=0$ one has
\begin{equation}\label{eq2:3}
    \overline{m}'(0)=-N,\quad \overline{m}''(0)=\frac{2N}{m_0-m_1}\,.
\end{equation}
In the case $m_0=m_1>m_2\geq \ldots \geq m_N$, i.e., when matrix $\bs{M}$ has the maximal eigenvalue of algebraic multiplicity 2, one has \cite{vishik1960solution}
\begin{equation}\label{eq2:7}
    \overline{m}'(0)=-N\pm\sqrt{N},
\end{equation}
i.e., the multiple eigenvalue of $\bs{M}$ splits into two simple ones.

We continue with a definition of the epsilon stabilization, which will be convenient to use for our approximate calculations:

\begin{definition}\label{def6:1}We shall say that the dominant eigenvalue $\overline{m}(\mu)$ of the problem \eqref{eq1:5} admits epsilon stabilization if for a small enough $\varepsilon>0$ there exist constant $\overline{m}^*_\varepsilon$ and such value of ${\mu}^*_\varepsilon$ that for all $\mu>\mu^*_\varepsilon$ the following conditions hold:
\begin{equation}\label{eq6:1}
    |\overline{m}(\mu)-\overline{m}^*_\varepsilon|<\varepsilon,\quad |\overline{m}'(\mu)|<\varepsilon.
\end{equation}
\end{definition}
This definition is weaker than the definition for asymptotic stabilization; its advantage is that it is more computationally oriented.

For the following it is convenient to rewrite the limit for the mean population fitness as
\begin{equation}\label{eq6:2}
    \hat{\overline{m}}=\lim_{\mu\to\infty}\overline{m}(\mu)=1+2^{-N}\sum_{k=0}^N(m_k-1)\binom{N}{k}.
\end{equation}
We also consider the scaling of the fitness vector $\bs{m}$ such that the minimum fitness is equal to 1: $m_0>m_1\geq m_2\geq\ldots\geq m_N=1$.

To find an approximate value for $\mu^*_\varepsilon$ we note that in the plane $\bigl(\mu,\overline{m}(\mu)\bigr)$ the condition
\begin{equation}\label{eq6:5}
\det\bigl(\bs{M}-\mu \bs{Q}_N-\overline{m}(\mu)\bs{I}\bigr)=0
\end{equation}
defines a curve, which we usually do not know explicitly, but which can be efficiently calculated using the parametric solution from Section \ref{sec5}. We shall call this curve \textit{the critical curve}. Moreover, due to Theorem \ref{th4:1} and asymptotic stabilization of the solution to \eqref{eq1:5}, this curve approaches the straight line $\overline{m}(\mu)\to \hat{\overline{m}}$.

Denote
$$
\delta=2^{-N}\sum_{k=0}^N (m_k-1)\binom{N}{k}\,.
$$
If $\delta<\varepsilon$ then the second term in the right hand side of \eqref{eq6:2} is negligible, which means that $\hat{\overline{m}}\approx\overline{m}^*=1$ in the plane $(\mu,\overline{m})$. Now using \eqref{eq2:3} we approximate the critical curve by a polynomial of the second degree, emanating from the point $(0,m_0)$:
\begin{equation}\label{eq6:6}
    \overline{m}_{app}=m_0-N\mu+\frac{N}{m_0-m_1}\mu^2.
\end{equation}
The curve \eqref{eq6:6} crosses the line $\hat{\overline{m}}=1$ at the point
\begin{equation}\label{eq6:3}
    \mu^*_\varepsilon=\frac{(m_0-m_1)}{2}\left(1-\sqrt{1-\frac{4(m_0-1)}{(m_0-m_1)N}}\right),\quad\mbox{if }\delta<\varepsilon.
\end{equation}
If $\delta\geq \varepsilon$ then the second term in \eqref{eq6:2} must be taken into account; i.e., we have that the critical curve should be $\varepsilon$ close to ${\overline{m}^*}=1+\delta$. Then the critical value of $\mu^*_\varepsilon$ is given by
\begin{equation}\label{eq6:4}
    \mu^*_\varepsilon=\frac{(m_0-m_1)}{2}\left(1-\sqrt{1-\frac{4(m_0-1)(\delta+1)}{(m_0-m_1)N}}\right),\quad\mbox{if }\delta\geq \varepsilon.
\end{equation}
To sum up, formulas \eqref{eq6:3} and \eqref{eq6:4} can be used as heuristic approximations for the critical mutation rate.

Expressions \eqref{eq6:3} and \eqref{eq6:4} can be simplified if instead of the curve \eqref{eq6:6} a straight line, starting at $(0,m_0)$, is taken:
$$
\overline{m}_{app}=m_0-N\mu.
$$
In this particular case
\begin{equation}\label{eq6:7}
    \mu^*_\varepsilon=\frac{m_0-1}{N}\,,\quad\mbox{if }\delta<\varepsilon,
\end{equation}
and
\begin{equation}\label{eq6:8}
    \mu^*_\varepsilon=\frac{m_0-(1+\delta)}{N}\,,\quad\mbox{if }\delta\geq\varepsilon.
\end{equation}

Finally, if matrix $\bs{M}$ has a maximal eigenvalue of multiplicity two then instead of \eqref{eq2:3} expression \eqref{eq2:7} should be taken into account.

\section{The error threshold}
Arguably the most important implication of Eigen's quasispecies theory is the presence of the so-called error threshold, which is not easy to define rigorously, but which can be described by the following famous example, which accompanies almost any discussion on the quasispecies model (and which we already have given a full analytical treatment in Section \ref{example:2}).

Consider again the Crow--Kimura quasispecies model \eqref{eq1:1} with the single peaked fitness landscape, i.e., $\bs M=\diag (m_0,m_1,\ldots,m_1),\,m_0>m_1$, and plot the stationary distribution of frequencies of different classes as the function of the mutation rate $\mu$. The result is shown in Fig.~\ref{fig:1}.
\begin{figure}[!b]
\includegraphics[width=\textwidth]{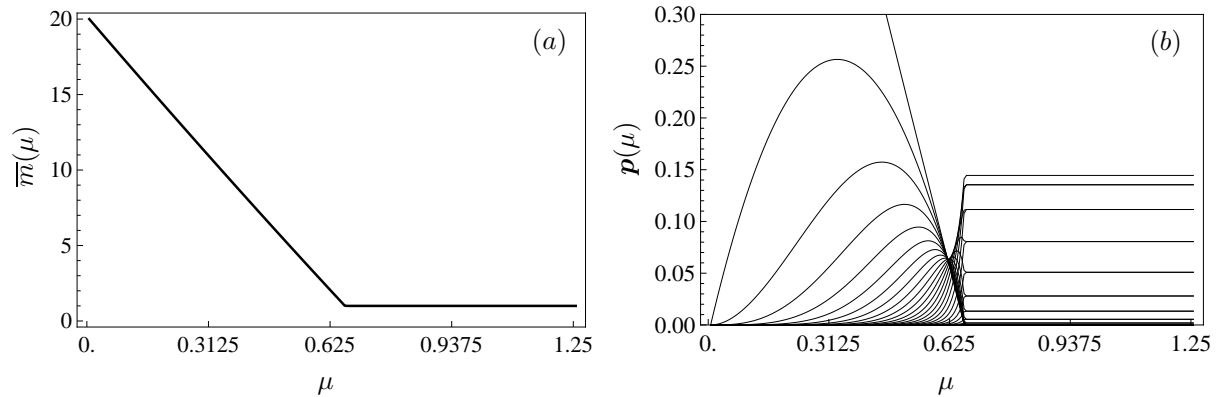}
\caption{Error threshold in the quasispecies model \eqref{eq1:1} with the single peaked fitness landscape ($\bs{m}=(m_0,m_1,\ldots,m_1),\,m_0>m_1$). The parameters are $N=30,\,m_0=20,\,m_1=1$. $(a)$ The mean population fitness $\overline{m}(\mu)$ versus the mutation rate; $(b)$ the stationary quasispecies distribution versus the mutation rate}\label{fig:1}
\end{figure}

A vague definition of the error threshold is that it separates two regimes of mutation--selection balance characterized by a qualitatively different structure of the principal eigenvector (quasispecies). For small mutation rates (see Fig.~\ref{fig:1}) the eigenvector is localized around the fittest sequence (the class that has fitness $m_0$). When the mutation rate is increased beyond the error threshold, the principal eigenvector becomes delocalized and the population spreads uniformly throughout the sequence space (for our particular model, as shown in Fig. \ref{fig:1}, ``spreads uniformly'' means that the distribution of the \textit{classes} of sequences is binomial). As Eigen and his co-authors write: ``Surpassing the threshold means melting of the quasi-species due to accumulation of errors. Such an error catastrophy means a sharp loss of genetic information.''~\cite{eigen1988mqs}

In Fig.~\ref{fig:1} we can observe a sharp transition between these two regimes; however, as we already discussed, if we deal with finite matrices, both eigenvalues and eigenvectors depend smoothly on its entries. Therefore, for the error threshold to become \textit{sharp} in the sense of producing some non-analytical behavior of the population distribution or the mean population fitness one needs the limit $N\to\infty$ (see Section \ref{example:2} for the exact procedure and results in the case of $N\to\infty$, when the leading eigenvalue $\overline{r}_\infty(\mu)=\max\{1-\mu,0\}$ is non-analytic at the point $\mu=1$). It is then necessary to rescale the parameters of the system, to observe in the limit what is called in statistical physics \textit{the phase transition} \cite{baake2001mutation}, and this can be taken as one of the rigorous definitions of the error threshold. In particular, in the left panel in Fig.~\ref{fig:1} the critical curve $\overline{m}(\mu)$ is shown; we proved in Section \ref{example:2}, as is also well known in the literature, that the leading eigenvalue has a jump of the first derivative. However, for any finite $N$, opposite to what our eye observes in the picture, this curve is smooth.

Another mathematically rigorous definition of the error threshold may be the critical mutation rate above which the distribution of the sequence classes becomes binomial (i.e., the distribution of the sequence types becomes uniform). However, again, as we showed in Proposition \ref{lem4:1}, this is possible for any finite $N$ only in the limit $\mu\to\infty$, and despite the fact that the distribution of the classes in Fig.~\ref{fig:1} looks very close to binomial, it slightly differs from the binomial distribution.

A detailed discussion of possible \textit{threshold-like behaviors} in the quasispecies model can be found in \cite{Hermisson2002}, where all exact definitions are based on infinite class limit $N\to\infty$ or on the classical approximation that ``the fittest type sequence becomes extinct.'' Here we would like to avoid any discussion of the infinite class limit, and also are not inclined to rely on the discussion of the \textit{extinction} phenomenon in the models that are formulated for the frequencies, as opposed to the absolute sizes (we note that the phenomenon of extinction has to be addressed in models with vital dynamics, as was stressed in, e.g., \cite{bull2007theory}).

There are several points to note about the error threshold.

First, as it is well known \cite{wiehe1997model}, not every fitness landscape produces the error threshold  defined vaguely as a \textit{sharp} transition. The classical example is the so-called Fujiyama or additive fitness landscape, in which the entries of the fitness vector $\bs m$ can be defined, e.g., as $m_k=k\log (1-s)$ for some constant $0<s<1$. The quasispecies vector versus the mutation rate is shown in Fig.~\ref{fig:2}.
\begin{figure}[!tbh]
\includegraphics[width=\textwidth]{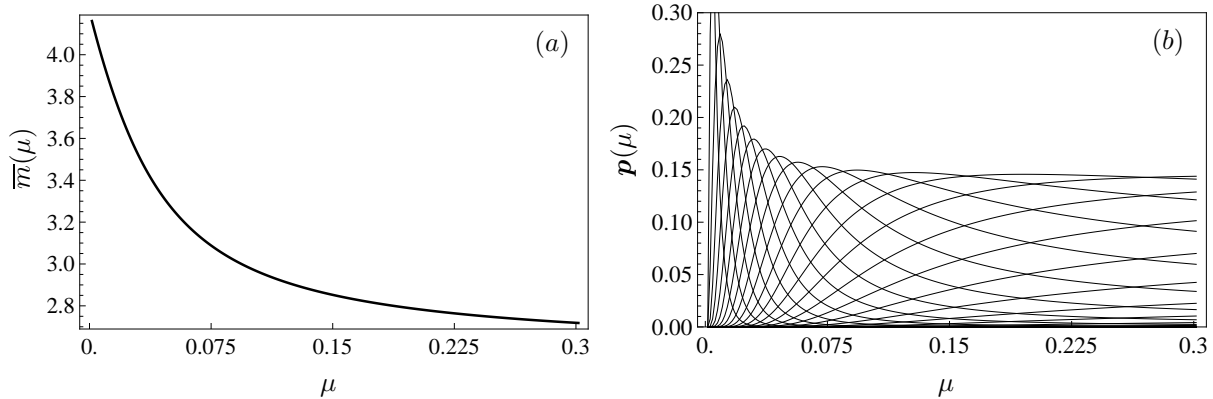}
\caption{Absence of the error threshold as a \textit{sharp} transition in the quasispecies model \eqref{eq1:1} with the additive fitness landscape $m_k=k\log (1-s),\,k=0,\ldots, N$. The parameters are $N=30,\,s=0.1$. $(a)$ The mean population fitness $\overline{m}(\mu)$ versus the mutation rate; $(b)$ the stationary quasispecies distribution versus the mutation rate}\label{fig:2}
\end{figure}

The second point is that given an arbitrary fitness landscape how one would estimate the critical mutation rate of the error threshold. A well known formula for the model \eqref{eq0:5} says that the critical mutation probability is inversely proportional to the sequence length and for the single peaked fitness landscape can be estimated as
$$
s^*=\frac{\log \frac{w_0}{w_1}}{N}\,.
$$
This approximation was found using the condition, which, provided the back mutations to the master sequence are prohibited, implies the extinction of this sequence.

For the Crow--Kimura model \eqref{eq1:1} and the single peaked fitness landscape the same approach yields
$$
\mu^*=\frac{m_0-m_1}{N}\,,
$$
which is exactly what we found by different methods in \eqref{eq6:7}. Although this formula works great for the single peaked landscape, it is expected that for many other fitness landscapes it will give a significant error (and it is not clear how to interpret its prediction for, e.g., additive fitness landscape, which lacks any \textit{sharp} transition).

Incidentally, the parametric solution we found in Section \ref{sec5} allows us to suggest another heuristic formula for the error threshold critical mutation rate. Ask the question: What do we see in Fig.~\ref{fig:1}? One of the pronounced features of this figure is that it looks like there exists a mutation rate such that at this rate all the frequencies are approximately the same. Hence, let us rephrase the question: For which matrices $\bs M$ all frequencies $p_k(\mu)$ pass for some $\mu^*$ through the same point? Obviously, we should have
\begin{equation}\label{eq7:2}
    p_0(\mu^*)=\ldots=p(\mu^*)=\frac{1}{N+1}\,,\quad \overline{m}(\mu^*)=\frac{1}{N+1}\sum_{k=0}^Nm_k.
\end{equation}
Straightforward computations show that the necessary and sufficient condition for $\bs M$ is
$$
\bs M=m \diag (1+r,1,\ldots,1,1+r)
$$
for some constants $m>0$ and $r>0$. Let $m=1$. Then we find
$$
\mu^*=\frac{r}{N+1}\,,\quad \overline m(\mu^*)=1+\frac{2r}{N+1}=1+2\mu^*.
$$
The curve $\bs p(\mu)$ at $\mu=\mu^*$ passes through the barycenter of $S_{N+1}$. Now consider another matrix $\bs M'$ which is close to $\bs M$. We can suppose (and the parametric solutions from Section \ref{sec5} show this as well) that the corresponding curve $\bs p(\mu)$ will pass close to the barycenter of $S_{N+1}$ and visually we will observe the picture that is very similar to the one in Fig. \ref{fig:1}. The results can be seen in Fig.~\ref{fig:3}.
%\begin{figure}[!t]
%\includegraphics[width=0.5\textwidth]{file5a.eps}
%\includegraphics[width=0.5\textwidth]{file5b.eps}
%\caption{Error threshold in the quasispecies model \eqref{eq1:1} with the single picked fitness landscape ($\bs{m}=(m_0,m_1,\ldots,m_1),\,m_0>m_1$). The parameters are $N=30,\,m_0=20,\,m_1=1$. Left panel: stationary quasispecies distribution $\bs x$ in the coordinates of the eigenvectors of $\bs Q$ versus mutation rate. Right panel: The same as on the left, in logarithmic coordinates}\label{fig:5}
%\end{figure}

\begin{figure}[!th]
\includegraphics[width=\textwidth]{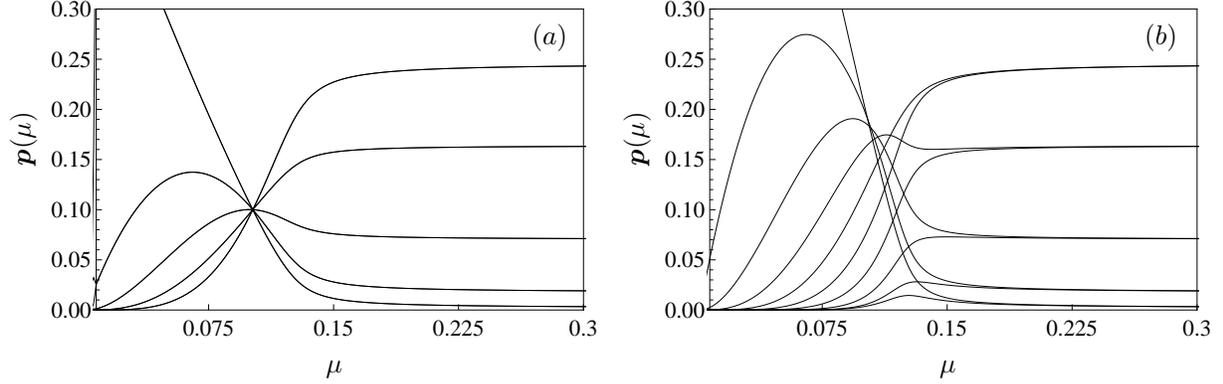}
\caption{Geometric view at the error threshold. $(a)$ The stationary quasispecies distribution versus mutation rate for $\bs{M}=\diag(1+r,1,\ldots,1,1+r),\,N=9,\,r=1$. $(b)$ The stationary quasispecies distribution versus mutation rate for $\bs{M}=\diag(1+r,1,\ldots,1,1+r'),\,N=9,\,r=1,\,r'=0.95$}\label{fig:3}
\end{figure}

A number of other examples (see Appendix \ref{append:C}) suggest the following formula to determine~$\mu^*$:
\begin{equation}\label{eq7:7}
    \overline{m}(\mu^*)=m_{\min} +2\mu^*.
\end{equation}
Since $\overline{m}(\mu^*)$ is the largest root of the equation $\det(\bs M+\mu^*\bs Q-\overline{m}\bs I)=0,$ then from condition~\eqref{eq7:7} an algebraic equation of degree $N+1$
\begin{equation}\label{eq7:8}
    \det(\bs M+\mu^* \bs Q-(m_{\min}+2\mu^*)\bs I)=0
\end{equation}
follows, from which critical $\mu^*$ should be determined.

For example, for the quasispecies model in Fig.~\ref{fig:1}, we find that using the linear approximation \eqref{eq6:7}, quadratic approximation \eqref{eq6:3}, and geometrically inspired formula \eqref{eq7:8} the values for $\mu^*$ are 0.633, 0.656, and 0.613 respectively. From Fig.~\ref{fig:1}$(a)$ we can estimate that the value of $\mu$ at which $\overline{m}(\mu)$ has a corner is $0.66$, which is better predicted by the quadratic approximation~\eqref{eq6:3}.

To give one more example, consider the fitness landscape of the form
\begin{equation}\label{eq7:9}
    m_k=k^\alpha\log (1-s),\quad k=0,\ldots,N,
\end{equation}
such that for $0<\alpha<1$ one has positive epistasis, and for $1<\alpha$ this fitness models negative epistasis, in case $\alpha=1$ we obtain again the additive fitness landscape. In Fig.~\ref{fig:4} an example of the quasispecies model with \eqref{eq7:9} is shown. The results of computations are $\mu^*=0.09$ by \eqref{eq6:7}, $\mu^*=0.22$ by \eqref{eq6:3} and $\mu^*=0.33$ by \eqref{eq7:8}.
\begin{figure}[!th]
\includegraphics[width=\textwidth]{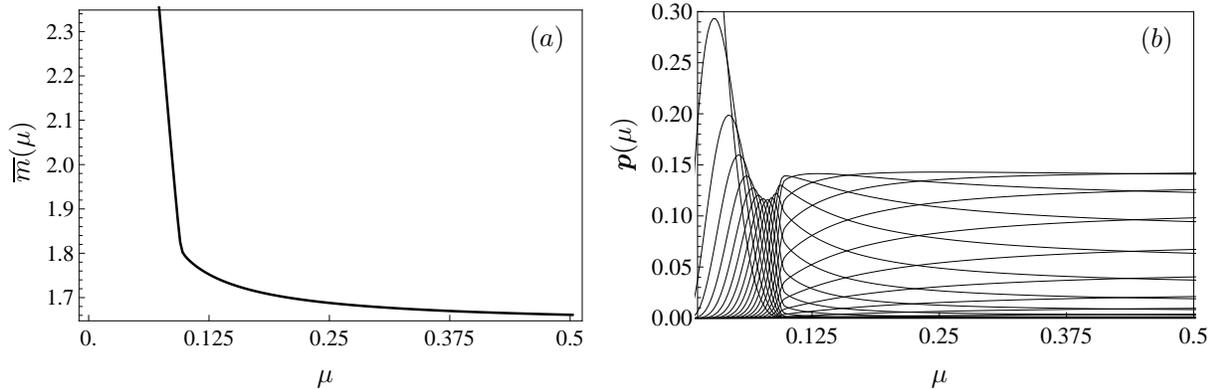}
\caption{Error threshold in the quasispecies model \eqref{eq7:9} with the positive epistasis; parameters are $N=30,\,s=0.7,\,\alpha=0.3$. $(a)$ The mean population fitness $\overline{m}(\mu)$ versus the mutation rate; $(b)$ the stationary quasispecies distribution versus the mutation rate}\label{fig:4}
\end{figure}

Numerical computations with other values of parameters confirm the following conclusion: When the fitness landscape is such that the graphs of $\bs p(\mu)$ do not cross at the same point, formula \eqref{eq7:8} gives a large error, whereas for the cases when  the frequencies of different classes pass close to a common point it gives accurate predictions; the linear approximation \eqref{eq6:7} points to the values of the mutation rate after which the structure of the quasispecies changes, and quadratic approximation \eqref{eq6:3} points to the point after which, in full accordance with the notion of epsilon stabilization, the quasispecies distribution of classes of sequences is close to the binomial one (and hence the types of sequences are distributed almost uniformly).

Summarizing, we would like to conclude this section with several points for a future discussion:
\begin{itemize}
\item The uniform distribution of the types of sequences is, in our opinion, the most convenient way to define the error threshold mathematically.
\item The uniform distribution is attained in the limit for any fitness landscape (Theorem \ref{th4:1}). Hence the error threshold, understood in the weak sense as a limit uniform distribution, is inherent in the quasispecies model and does not depend on the fitness landscape, contrary to the fact that the phase transition phenomenon (as well as other possible threshold-like transitions) depends on the fitness landscape.
\item Since for the finite sequence length the distribution is never exactly uniform (except for the trivial case of the scalar matrix $\bs{M}$), the definition of the error threshold should include the degree of closeness to the uniform distribution (or to the binomial distribution if we speak of the classes of sequences). Therefore, Definition \ref{def6:1} of the epsilon stabilization should be used as a mathematically rigorous definition of the error threshold.
\item Various approximate formulas are possible to obtain to estimate the critical value of the error threshold defined through the notion of epsilon stabilization (see Section \ref{ep_stab} and \eqref{eq7:8}). Among other things, we mention that not all of them show that the critical mutation rate is inversely proportional to the sequence length, as usually implied in biologically oriented discussions.
\end{itemize}

\appendix

%\section{Additional properties of matrix $\bs C$}\label{append:A}
%Here we list several additional interesting properties of matrix $\bs C$ without proof.
%\begin{itemize}
%\item All $c_{ik}$ are integers with the following symmetries
% $$c_{N-i,\, k}=(-1)^k c_{i k}\;,\quad c_{i,\, N-k}=(-1)^i c_{i k}\quad i=0,\ldots,N,\, k=0,\ldots,N.$$
%
%\item $\det \bs C=(-2)^{N(N+1)/2}$.
%
%\item The Gramian matrix $\bs G=\bs C^\top\bs C$ (i.e., the matrix of the inner products
%$\bs G=(g_{ij})=(\bs v_i\cdot \bs v_j)$) has the entries
%$$
%g_{ij}=\left\{ \begin{array}{cl}\vspace{3pt} 0\;,&\mbox{if
%$i+j=2a+1$ is odd,}\\
%\displaystyle{(-1)^{(i-j)/2}\, \frac{\binom{2a}{a}
%\binom{2(N-a)}{N-a}}{\binom{N}{a}}}\;,&\mbox{if $i+j=2a$ is even}\,.\\\end{array} \right.
%$$
%Therefore the eigenvectors $\bs v_i$ и $\bs v_j$ for indexes of different oddness are orthogonal.
%
%\item For the sum of the elements of the row $i$ we have
%$$
%\sum_{k=0}^N c_{ik}=\left\{ \begin{array}{cl}\vspace{3pt}
%0\;,&\mbox{if $i=2a+1$ is odd,}\\
%\binom{N+1}{i+1}\;,&\mbox{if $i=2a$ is even}\,.\\\end{array}
%\right.
%$$
%\item For any $0\leq i,j\leq N$
%$$
%c_{ij}{N\choose j}=c_{ji}{N\choose i}.
%$$
%\end{itemize}

\section{Solutions to system \eqref{eq1:5}} \label{append:B}
In the main text we presented a parametric solution \eqref{eq5:10}, which is valid only in the case when the fitness landscape has one non-zero entry. Here we present a general solution to system \eqref{eq1:5}.
\subsection{Solution to system \eqref{eq1:5} in case of two positive fitnesses}
Let $m_{j_1}>0,\,m_{j_2}>0$ and all other $m_i=0$. Then, using \eqref{eq5:5},
\begin{equation}\label{eq5:14}
   {F}(s)=\overline{m}=m_{j_1}p_{j_1}(s)+m_{j_2}p_{j_2},
\end{equation}
and function ${F}(s)$ is positive at least for small $s$ since ${F}(s)\to m_{\max}$ as $s\to 0$. System \eqref{eq5:4} takes the form
\begin{equation}\label{eq5:15}
\begin{split}
    m_{j_1}p_{j_1}&=\overline{m}\sum_{k=0}^Nc_{{j_1},k}(1+ks)x_k,\quad    m_{j_2}p_{j_2}=\overline{m}\sum_{k=0}^Nc_{{j_2},k}(1+ks)x_k,\\
   0&=\overline{m}\sum_{k=0}^Nc_{ik}(1+ks)x_k,\quad i=0,\ldots, N,\,i\neq j_1,\,i\neq j_2,\quad   p_i=\sum_{k=0}^Nc_{ik}x_k,\quad i=0,\ldots,N,
\end{split}
\end{equation}
or, after dividing by ${F}(s)=\overline{m}$,
\begin{equation}\label{eq5:16}
\begin{split}
    \frac{m_{j_1}p_{j_1}}{F(s)}&=\sum_{k=0}^Nc_{{j_1},k}(1+ks)x_k,\quad     \frac{m_{j_2}p_{j_2}}{F(s)}=\sum_{k=0}^Nc_{{j_2},k}(1+ks)x_k,\\
   0&=\sum_{k=0}^Nc_{ik}(1+ks)x_k,\quad i=0,\ldots, N,\,i\neq j_1,\,i\neq j_2,\quad   p_i=\sum_{k=0}^Nc_{ik}x_k,\quad i=0,\ldots,N.
\end{split}
\end{equation}
In the matrix form \eqref{eq5:16} is
$$
\frac{m_{j_1}p_{j_1}}{F(s)}\bs{e}_{j_1}+\frac{m_{j_2}p_{j_2}}{F(s)}\bs{e}_{j_2}=\bs{C}\diag(1,1+s,1+2s,\ldots,1+Ns)\bs x,\quad \bs{p}=\bs{Cx}.
$$
Using the fact that $\bs{C}=2^{-N}\bs{C}$, we obtain
$$
\diag(1,1+s,1+2s,\ldots,1+Ns)\bs x=2^{-N}\left(\frac{m_{j_1}p_{j_1}}{F(s)}\bs{Ce}_{j_1}+\frac{m_{j_2}p_{j_2}}{F(s)}\bs{Ce}_{j_2}\right)
$$
or, in coordinates, using \eqref{eq5:11}:
\begin{equation}\label{eq5:17}
\begin{split}
    x_k(s)&=\frac{1}{2^N}\frac{c_{k,{j_1}}m_{j_1}+c_{k,{j_2}}m_{j_2}}{F(s)(1+ks)}, \quad p_k(s)=\frac{m_{j_1}p_{j_1}(s)}{F(s)}F_{k,j_1}(s)+\frac{m_{j_2}p_{j_2}(s)}{F(s)}F_{k,j_2}(s).
\end{split}
\end{equation}
In system \eqref{eq5:17} put $j_1$ and $j_2$ in equations for $p_k$:
\begin{equation}\label{eq5:18}
\begin{split}
F(s)p_{j_1}(s)=m_{j_1}F_{j_1,j_1}(s)p_{j_1}+m_{j_2}F_{j_1,j_2}(s)p_{j_2},\\
F(s)p_{j_2}(s)=m_{j_1}F_{j_2,j_1}(s)p_{j_1}+m_{j_2}F_{j_2,j_2}(s)p_{j_2}
\end{split}
\end{equation}
System \eqref{eq5:18} has to have a nontrivial solution (at least for $s$ close to 0), for which it is necessary and enough that
\begin{equation}\label{eq5:19}
 \det\begin{bmatrix}
       m_{j_1}F_{j_1,j_1}(s)-F(s) & m_{j_2}F_{j_1,j_2}(s) \\
       m_{j_1}F_{j_2,j_1}(s)& m_{j_2}F_{j_2,j_2}(s)-F(s) \\
     \end{bmatrix}=0,
 \end{equation}
or,
$$
F^2(s)-\Tr(j_1,j_2)(s)+D(j_1,j_2)(s)=0,
$$
where
$$
\Tr(j_1,j_2)(s)=m_{j_1}F_{j_1,j_1}(s)+m_{j_2}F_{j_2,j_2}(s),
$$
$$
D(j_1,j_2)(s)=m_{j_1}m_{j_2}(F_{j_1,j_1}(s)F_{j_2,j_2}(s)-F_{j_1,j_2}(s)F_{j_2,j_1}(s)).
$$
If $m_{j_2}=0$ we get that $D(j_1,j_2)=0$, which means that one of the roots of the quadratic equation is $F(s)=0$, which is not interesting for us. Therefore, we need to choose the root, which is given by
\begin{equation}\label{eq5:20}
    F(s)=\frac 12\left(\Tr(j_1,j_2)(s)+\sqrt{\Tr^2(j_1,j_2)(s)-4D(j_1,j_2)(s)}\right).
\end{equation}
Since we have $F(s)$, we can use one of the equations in \eqref{eq5:18} and adding \eqref{eq5:14}, we can find $p_{j_1}$ and $p_{j_2}$. Other $p_i$ are found by \eqref{eq5:17}. The final answer is
$$
p_{j_1}(s)=\frac{F(s)F_{j_1,\,j_2}(s)}{F(s)+m_{j_1}(F_{j_1,\,j_2}(s)-F_{j_1,\,j_1}(s))}\,,\quad
p_{j_2}(s)=\frac{F(s)(F(s)-m_{j_1}\,F_{j_1,\,j_1}(s))}{m_{j_2}(F(s)+m_{j_1}(F_{j_1,\,j_2}(s)-F_{j_1,\,j_1}(s)))}\,,
$$
$$p_i(s)= \frac{m_{j_1}\,
F_{i,\,j_1}(s) F_{j_1,\,j_2}(s)+F_{i,\,j_2}(s)(F(s)-m_{j_1}\,F_{j_1,\,j_1}(s))}{F(s)+m_{j_1}(F_{j_1,\,j_2}(s)-F_{j_1,\,j_1}(s))}\;;\quad
i\ne j_1\,,j_2\,.$$
\begin{remark}
In the case $m_{j_1}=m_{j_2}$ in the last expressions it is necessary to apply limit for $s\to0^+$.
\end{remark}
\subsection{Solution to system \eqref{eq1:5} in case of several positive fitnesses}
Let for indexes $0\leq {j_1}<\ldots <{j_l}\leq N$ it it true that $m_{j_k}>0,\,k=1,\ldots,l$, and for all other indexes $m_i=0$; assume that $m_i=0$ at least for one $i$. Generalizing the previous reasoning, we have
\begin{equation}\label{eq5:21}
    F(s)=\overline{m}=\sum_{k=1}^lm_{j_k}p_{j_k}(s),\quad \mu=\frac 12 sF(s).
\end{equation}
Function $F(s)$ is the largest root of the algebraic equation of degree $l$ with non-constant coefficients
\begin{equation}\label{eq5:22}
    \det(m_{j_k}F_{j_i,j_k}(s)-F\bs{I})=0.
\end{equation}
If among nonzero $m_{j_k}$ there is the biggest one, this root is simple (at least for small $s$). The solutions $p_{j_k}$ can be found as follows. In the system
$$
F(s)p_{j_i}(s)=\sum_{k=1}^lm_{j_k}F_{j_i,j_k}(s)p_{j_k}(s),\quad i=1,\ldots, l,
$$
one of the equations is changed for \eqref{eq5:21} and this system is solved for $p_{j_k}$. For all other $p_i$ such that $m_i=0$ the formulas
$$
p_i(s)=\sum_{k=1}^l\frac{m_{j_k}p_{j_k}(s)}{F(s)}F_{i,j_k}(s)
$$
are used.

\section{Proof of Lemma \ref{lemQ:1}}\label{append:D}
\begin{proof}[Proof of Lemma \ref{lemQ:1}] In this proof we will need to change the dimension of the problem, therefore, we introduce the notation $\bs C=\bs C^{(N)}=(c_{ka}^{(N)})$ and
$$
p_a^{(N)}=\frac{{N\choose a}}{2^{N}}\sum_{k=0}^{N}
\frac{c_{ka}^{(N)}}{1+2u\frac{k}{N}}\;,\quad 0\leq a\leq
N.
$$
We apply Abel's transformation to the last sum and obtain
$$
\sum_{k=0}^{N} c_{ka}^{(N)}x_k=\sum_{k=0}^{N-1}
S_{ka}^{(N)}(x_k-x_{k+1})+S_{Na}^{(N)}x_N\,;\quad
S_{ka}^{(N)}=\sum_{j=0}^{k}c_{ja}^{(N)}\;,\;\;x_k=\frac{1}{{1+2u\frac{k}{N}}}\,.
$$
We have $S_{Aa}^{(N)}=0$ and
$$
S_{ka}^{(N)}=c_{k,a-1}^{(N-1)}.
$$
Indeed, sum $S_{ka}^{(N)}$ is obtained as the coefficient at $t^k$ in the generating function
$$
(1+t+t^2+\ldots)(1-t)^a(1+t)^{N-a}=(1-t)^{a-1}(1+t)^{N-a}.
$$
Therefore, after transformation the expression for $p_a^{(N)}$ takes the form
$$
p_a^{(N)}=\frac{{N\choose a} u}{N
2^{N-1}}\sum_{k=0}^{N-1}
\frac{c_{k,a-1}^{(N-1)}}{\left(1+2u\frac{k}{N}\right)
\left(1+2u\frac{k+1}{N}\right)}\,.
$$
By applying Abel's transformation in the same spirit, after $a-1$ steps we obtain \eqref{eqQ:14}, which implies \eqref{eqQ:15}, which needs to be proved.

Consider the difference
$$
 \frac{1}{ 2^{N-a}}\sum_{k=0}^{N-a} \frac{{N-a\choose
k}}{\prod\limits_{j=k}^{k+a}\left(1+\frac{2u
j}{N}\right)}-\frac{1}{(1+u)^{a+1}}= \frac{1}{
2^{N-a}}\sum_{k=0}^{N-a} \!{N-a\choose k}\!\!\left(
\frac{1}{\prod\limits_{j=k}^{k+a}\left(1+\frac{2u
j}{N}\right)}- \frac{1}{(1+u)^{a+1}}\right).
$$
We note that for any numbers, and in particular for $x_j=\frac{1}{1+2u\frac{j}{N}}$ and $y=\frac{1}{1+u}$, it is true that
$$
x_kx_{k+1}\ldots x_{k+a}-y^{a+1}=\sum_{c=0}^a y^c (x_{k+c}-y)x_{k+c+1}x_{k+c+2}\ldots x_{k+a}.
$$
Applying the last expression to the difference above and using the fact that
$$
x_j-y=\frac{1}{1+2u\frac{j}{N}}-\frac{1}{1+u}=
\frac{u(N-2j)}{N\left(1+2u\frac{j}{N}\right)(1+u)}\,,
$$
we find
\begin{equation}\label{eqD:18}
\Delta^{(N)}_a=\frac{1}{ 2^{N-a}}\sum_{k=0}^{N-a} \frac{{N-a\choose
k}}{\prod\limits_{j=k}^{k+a}\left(1+\frac{2u
j}{N}\right)}-\frac{1}{(1+u)^{a+1}}= \frac{u}{ N
2^{N-a}}\sum_{k=0}^{N-a} {N-a\choose k}\sum_{j=k}^{k+a}
(N-2j)T_{kj},
\end{equation}
where the explicit expressions for the constants $T_{kj}$ are not important because all we need is $0\leq T_{kj}\leq 1$ being the product of numbers from $[0,1]$.

In the sum \eqref{eqD:18} there are equal numbers of positive ($2j<N$) and negative ($2j>N$) terms. Regroup the terms and find
$$
\Delta^{(N)}_a= \frac{u}{ N 2^{N-a}}\sum_{j=0}^{N}
(N-2j) \sum_{k=j-a}^{j} {N-a\choose k}
T_{kj}.
$$
Now we notice that
$$
\frac{{N\choose a} a! }{N^a}\Delta^{(N)}_a= \frac{{N\choose
a} a! }{N^a}\Delta^{(N)}_{a,+}-\frac{{N\choose a} a!
}{N^a}\Delta^{(N)}_{a,-}=:\Sigma_+-\Sigma_-,
$$
where
\begin{align*}
\Delta^{(N)}_{a,+}&=\frac{u}{ N 2^{N-a}}\sum_{j=0}^{L}
(N-2j) \sum_{k=j-a}^{j} {N-a\choose k} T_{kj},\\
\Delta^{(N)}_{a,-}&=\frac{u}{ N
2^{N-a}}\sum_{j=L+1}^{N} (2j-N) \sum_{k=j-a}^{j} {N-a\choose
k} T_{kj},\quad L=[N/2].
\end{align*}
Both sums $\Sigma_+$ and $\Sigma_-$ are non-negative and, due to the symmetry of the binomial coefficients, can be bounded by the same number depending on $N$. Here is an estimate for $\Sigma_+$:
$$
\Sigma_+\leq\frac{{N\choose a} a!
}{N^{a+1}}\frac{u}{ 2^{N-a}}\sum_{j=0}^{L} (N-2j)
\sum_{k=j-a}^{j} {N-a\choose k}\leq\frac{u\cdot 2^a}{
N^{a+1} 2^{N}}\sum_{j=0}^{L} (N-2j) \sum_{k=j-a}^{j}\!\!a!
{N\choose a} {N-a\choose k}.
$$
We also have
$$
\sum_{k=j-a}^{j}\frac{a! {N\choose a} {N-a\choose
k}}{N^a}=\sum_{k=j-a}^{j}{N\choose
k}\frac{(N-k)!}{(N-a-k)!
N^a}\leq\sum_{k=j-a}^{j}{N\choose k}\leq(a+1){N\choose
j}\,,
$$
because the binomial coefficients increase for $k\leq j\leq L=[N/2]$. This implies that
$$
\Sigma_+\leq\frac{u 2^a(a+1)}{ N
2^{N}}\sum_{j=0}^{L} (N-2j){N\choose j}\,.
$$
Now, using (see \eqref{eq3:6new}) $c_{ak}{N\choose k}=c_{ka}{N \choose a}$ and $S_{ka}^{(N)}=c_{k,a-1}^{(N-1)}$, we obtain
$$
\frac{1}{2^{N}}\sum_{j=0}^{L} (N-2j){N\choose
j}=\frac{1}{2^{N}}\sum_{j=0}^{L} c^{(N)}_{1,j}{N\choose j}=
\frac{1}{2^{N}}\sum_{j=0}^{L} c^{(N)}_{j,1}{N\choose
1}=\frac{N}{2^{N}}\sum_{j=0}^{L}
c^{(N)}_{j,1}=\frac{N}{2^{N}}{N-1\choose L}\,.
$$
By using the fact
$$
\frac{1}{2^{N-1}}{N-1\choose L}<\frac{1}{\sqrt{N}}\,,
$$
we finally arrive at
$$
\Sigma_+\leq\frac{u 2^a(a+1)}{ N
2^{N}}\sum_{j=0}^{L} (N-2j){N\choose j}<\frac{u
2^a(a+1)}{2\sqrt{N}}\,,
$$
and
\begin{equation}\label{eqD:25}
\frac{{N\choose a} a!
}{N^a}|\Delta^{(N)}_a|=|\Sigma_+-\Sigma_-|\leq
2\Sigma_+\leq\frac{u 2^a(a+1)}{\sqrt{N}}\,.
\end{equation}
Equations \eqref{eqD:18} and \eqref{eqD:25} yield
\begin{align*}
\left|p_a-\frac{u^a}{(1+u)^{a+1}}\right|&\le \left|p_a-\frac{{N\choose a} a!}{N^a}\frac{u^a}{(1+u)^{a+1}}\right|+\left(1-\frac{{N\choose a} a! }{N^a}\right)\frac{u^a}{(1+u)^{a+1}}\\
&\le \frac{{N\choose a} a!\,|\Delta_a^{N}| u^a} {N^a}+\left(1-\frac{{N\choose a} a!}{N^a}\right)\frac{u^a}{(1+u)^{a+1}}\\
&\le\frac{2^a(a+1)u^{a+1}}{\sqrt{N}}+\left(1-\frac{{N\choose a}\cdot a!}{N^a}\right)\frac{u^a}{(1+u)^{a+1}}\;,
\end{align*}
which concludes the proof of Lemma \ref{lemQ:1}.
\end{proof}

\section{Geometric view at the error threshold}\label{append:C}
In the main text we found that the critical mutation rate such that all the frequencies of the quasispecies distribution pass through the barycenter of $S_N$ can be found as $\overline{m}(\mu^*)=1+2\mu^*$ in the case of $\bs{M}=\diag(1+r,1,\ldots,1,1+r)$. As a next example consider matrix
$$
\bs M=m\diag(1,1+r,1,\ldots,1,1+r,1).
$$
Let $m=1$. The curve $\bs p(\mu)$ at
$$
\mu^*=\frac{r(N+2)}{N(N+1)}\,,\quad \overline{m}(\mu^*)=1+\frac{2r (N+2)}{N(N+1)}=1+2\mu^*,
$$
passes through the point
\begin{equation}\label{eq7:4}
    b_1=\frac{1}{N(N+1)}(1,N+2,N+2,\ldots,N+2,N+2,1)^\top
\end{equation}
of the simples $S_N$. Therefore $p_1(\mu^*)=\ldots=p_{N-1}(\mu^*)$ and $N-2$ graphs cross at $b_1$. For close matrices $\bs M'$ we shall also observe picture similar to the ``error threshold'' in Fig.~\ref{fig:1}. Note that since $b_1\neq b_0$ then \eqref{eq7:2} is not exact in this case.

Similarly, for the matrix
$$
\bs M=m\diag(1,1,1+r,1,\ldots,1,1+r,1,1)
$$
with $m=1$ we find that at
$$
\mu^*=\frac{r(N^2+3N+4)}{(N-1)N(N+1)},\quad \overline{m}(\mu^*)=1+\frac{2r(N^2+3N+4)}{(N-1)N(N+1)}=1+2\mu^*,
$$
the curve $\bs p(\mu)$ passes through the point
$$
b_2=\frac{2}{(N+1)N(N-1)}\Bigl(1,N+2,(N^2+3N+4)2^{-1},\ldots,(N^2+3N+4)2^{-1}, N+2,1\Bigr)^\top
$$
of $S_N$. Hence $p_2(\mu^*)=\ldots=p_{N-2}(\mu^*)$ and $N-4$ graphs cross at the same point; for close matrices $\bs M'$ for $\mu\approx \mu^*$ again the picture of the error threshold will be observed. The process can be continued for the points $b_0,b_1,b_2,b_3,\ldots$.

It is not necessary that the fitness values are symmetric as in all the cases above. Consider the matrix
$$
\bs{M}=\diag (1,1+\frac{rN}{N+2},1\ldots,1,1+r).
$$
We have that at
$$
\mu^*=\frac{r}{N+1},\quad \overline{m}(\mu^*)=1+2\mu^*=1+\frac{2r}{N+1}\,,
$$
the curve $\bs p(\mu)$ passes through the point
$$
b=\frac{1}{(N+1)}(1,N+2,\ldots,N+2)^\top,
$$
which means that $N$ out of $N+1$ graphs pass through the same point.

%Summing the results together we obtain that in every considered case, when the curve $\bs p(\mu)$ passes through some ``center'' point of the simplex $S_N$ (i.e., the point where many of the coordinates are the same), the threshold value of $\mu$ can be found from the relation $\overline(\mu^*)=1+2\mu^*$. In the general case for arbitrary matrix $\bs M$ we suggest the following formula to determine $\mu^*$:
%\begin{equation}\label{eq7:7}
%    \overline{m}(\mu^*)=m_{\min} +2\mu^*.
%\end{equation}
%Since $\overline{m}(\mu^*)$ is the largest root of the equation $\det(\bs M+\mu^*\bs Q-\overline{m}\bs I)=0,$ then from condition \eqref{eq7:7} an algebraic equation of degree $N+1$
%\begin{equation}\label{eq7:8}
%    \det(\bs M+\mu^* \bs Q-(m_{\min}+2\mu^*)\bs I)=0.
%\end{equation}
%follows, from which the critical $\mu^*$ should be determined.

\paragraph{Acknowledgements:} This research is supported in part by the Russian Foundation for Basic Research (RFBR) grant \#10-01-00374 and joint
grant between RFBR and Taiwan National Council \#12-01-92004HHC-a. ASN's research is supported in part by ND EPSCoR and NSF grant \#EPS-0814442.

%\bibliography{prebiotic}

\begin{thebibliography}{10}

\bibitem{baake1997ising}
E.~Baake, M.~Baake, and H.~Wagner.
\newblock Ising quantum chain is equivalent to a model of biological evolution.
\newblock {\em Physical Review Letters}, 78(3):559--562, 1997.

\bibitem{baake1999}
E.~Baake and W.~Gabriel.
\newblock {Biological evolution through mutation, selection, and drift: An
  introductory review}.
\newblock In D.~Stauffer, editor, {\em {Annual Reviews of Computational Physics
  VII}}, pages 203--264. World Scientific, 1999.

\bibitem{Baake2007}
E.~Baake and H.-O. Georgii.
\newblock Mutation, selection, and ancestry in branching models: a variational
  approach.
\newblock {\em Journal of Mathematical Biology}, 54(2):257--303, Feb 2007.

\bibitem{baake2001mutation}
E.~Baake and H.~Wagner.
\newblock Mutation--selection models solved exactly with methods of statistical
  mechanics.
\newblock {\em Genetical research}, 78(1):93--117, 2001.

\bibitem{bull2007theory}
J.~J. Bull, R.~Sanjuan, and C.~O. Wilke.
\newblock Theory of lethal mutagenesis for viruses.
\newblock {\em Journal of virology}, 81(6):2930--2939, 2007.

\bibitem{burger2000mathematical}
R.~B{\"{u}}rger.
\newblock {\em {The mathematical theory of selection, mutation, and
  recombination}}.
\newblock Wiley, 2000.

\bibitem{crow1970introduction}
J.~F. Crow and M.~Kimura.
\newblock {\em An introduction to population genetics theory.}
\newblock New York, Evanston and London: Harper \& Row, Publishers, 1970.

\bibitem{eigen1971sma}
M.~Eigen.
\newblock Selforganization of matter and the evolution of biological
  macromolecules.
\newblock {\em Naturwissenschaften}, 58(10):465--523, 1971.

\bibitem{eigen1989mcc}
M.~Eigen, J.~McCascill, and P.~Schuster.
\newblock {The Molecular Quasi-Species}.
\newblock {\em Advances in Chemical Physics}, 75:149--263, 1989.

\bibitem{eigen1988mqs}
M.~Eigen, J.~McCaskill, and P.~Schuster.
\newblock Molecular quasi-species.
\newblock {\em Journal of Physical Chemistry}, 92(24):6881--6891, 1988.

\bibitem{galluccio1997exact}
S.~Galluccio.
\newblock Exact solution of the quasispecies model in a sharply peaked fitness
  landscape.
\newblock {\em Physical Review E}, 56(4):4526, 1997.

\bibitem{garcia2002linear}
R.~Garc{\i}a-Pelayo.
\newblock A linear algebra model for quasispecies.
\newblock {\em Physica A: Statistical Mechanics and its Applications},
  309(1):131--156, 2002.

\bibitem{Hermisson2002}
J.~Hermisson, O.~Redner, H.~Wagner, and E.~Baake.
\newblock Mutation-selection balance: ancestry, load, and maximum principle.
\newblock {\em Theoretical Population Biology}, 62(1):9--46, Aug 2002.

\bibitem{hofbauer1985selection}
J.~Hofbauer.
\newblock The selection mutation equation.
\newblock {\em Journal of Mathematical Biology}, 23(1):41--53, 1985.

\bibitem{jainkrug2007}
K.~Jain and J.~Krug.
\newblock {Adaptation in Simple and Complex Fitness Landscapes}.
\newblock In U.~Bastolla, M.~Porto, H.~Eduardo~Roman, and M.~Vendruscolo,
  editors, {\em Structural approaches to sequence evolution}, chapter~14, pages
  299--339. Springer, 2007.

\bibitem{karev2010}
G.~P. Karev, A.~S. Novozhilov, and F.~S. Berezovskaya.
\newblock On the asymptotic behavior of the solutions to the replicator
  equation.
\newblock {\em Mathematical Medicine and Biology}, 28(2):89--110, 2011.

\bibitem{karlin1965ehrenfest}
S.~Karlin and J.~McGregor.
\newblock Ehrenfest urn models.
\newblock {\em Journal of Applied Probability}, 2(2):352--376, 1965.

\bibitem{kato1995perturbation}
T.~Kat{\=o}.
\newblock {\em Perturbation theory for linear operators}, volume 132.
\newblock Springer Verlag, 1995.

\bibitem{leuthausser1986exact}
I.~Leuth{\"a}usser.
\newblock {An exact correspondence between Eigen's evolution model and a
  two-dimensional Ising system}.
\newblock {\em The Journal of Chemical Physics}, 84(3):1884--1885, 1986.

\bibitem{leuthausser1987statistical}
I.~Leuth{\"a}usser.
\newblock {Statistical mechanics of Eigen's evolution model}.
\newblock {\em Journal of statistical physics}, 48(1):343--360, 1987.

\bibitem{nowak1989error}
M.~Nowak and P.~Schuster.
\newblock {Error thresholds of replication in finite populations mutation
  frequencies and the onset of Muller's ratchet}.
\newblock {\em Journal of Theoretical Biology}, 137(4):375--395, 1989.

\bibitem{rellich1969perturbation}
F.~Rellich.
\newblock {\em Perturbation theory of eigenvalue problems}.
\newblock Routledge, 1969.

\bibitem{rumschitzki1987spectral}
D.~S. Rumschitzki.
\newblock {Spectral properties of Eigen evolution matrices}.
\newblock {\em Journal of Mathematical Biology}, 24(6):667--680, 1987.

\bibitem{saakian2004ese}
D.~B. Saakian, C.~K. Hu, and H.~Khachatryan.
\newblock{Solvable biological evolution models with general fitness functions and multiple mutations in parallel mutation-selection scheme}.
\newblock{\em Physical Review E},
 70(4): 041908, 2004.

\bibitem{saakian2006ese}
D.~B. Saakian and C.~K. Hu.
\newblock {Exact solution of the Eigen model with general fitness functions and
  degradation rates}.
\newblock {\em Proceedings of the National Academy of Sciences USA},
  103(13):4935--4939, 2006.

\bibitem{schuster1988stationary}
P.~Schuster and J.~Swetina.
\newblock Stationary mutant distributions and evolutionary optimization.
\newblock {\em Bulletin of Mathematical Biology}, 50(6):635--660, 1988.

\bibitem{semenov2014}
Y.~S. Semenov, A.~S. Bratus, and A.~S. Novozhilov.
\newblock {On the behavior of the leading eigenvalue of the Eigen evolutionary
  matrices}.
\newblock page~{\it in preparation}, 2014.

\bibitem{swetina1982self}
J.~Swetina and P.~Schuster.
\newblock Self-replication with errors: A model for polvnucleotide replication.
\newblock {\em Biophysical Chemistry}, 16(4):329--345, 1982.

\bibitem{tejero2011relationship}
H.~Tejero, A.~Mar{\'\i}n, and F.~Montero.
\newblock The relationship between the error catastrophe, survival of the
  flattest, and natural selection.
\newblock {\em BMC Evolutionary Biology}, 11(1):2, 2011.

\bibitem{vishik1960solution}
M.~I. Vishik and L.~A. Lyusternik.
\newblock {The solution of some perturbation problems for matrices and
  selfadjoint or non-selfadjoint differential equations I}.
\newblock {\em Russian Mathematical Surveys}, 15(3):1--73, 1960.

\bibitem{wiehe1997model}
T.~Wiehe.
\newblock Model dependency of error thresholds: the role of fitness functions
  and contrasts between the finite and infinite sites models.
\newblock {\em Genetical research}, 69(02):127--136, 1997.

\bibitem{wilke2005quasispecies}
C.~O. Wilke.
\newblock Quasispecies theory in the context of population genetics.
\newblock {\em BMC Evolutionary Biology}, 5(1):44, 2005.

\end{thebibliography}

\end{document}